\documentclass[journal,twoside,web]{ieeecolor}

\usepackage{generic}
\usepackage{cite}
\usepackage{amsmath,amssymb,amsfonts}
\usepackage{multirow}
\usepackage{booktabs}
\usepackage{caption}
\usepackage{algorithm}
\usepackage{algpseudocode}
\usepackage{graphicx}
\usepackage{textcomp}
\usepackage{scalerel}
\usepackage{subfigure}

\usepackage{amsthm}

\def\BibTeX{{\rm B\kern-.05em{\sc i\kern-.025em b}\kern-.08em
    T\kern-.1667em\lower.7ex\hbox{E}\kern-.125emX}}
\begin{document}
\title{Sparse Polynomial Zonotopes: A Novel Set Representation for Reachability Analysis}

\author{Niklas Kochdumper, Matthias Althoff
\thanks{This paragraph of the first footnote will contain the date on 
which you submitted your paper for review. This work was supported by the German Research Foundation (DFG) project faveAC under grant number AL 1185/5-1 }
\thanks{Niklas Kochdumper and Matthias Althoff are both with the Department of Computer Science, Technical University of Munich, 85748 Garching, Germany (e-mail: niklas.kochdumper@tum.de, althoff@tum.de). }}

% Caption format
\captionsetup{
    width=\linewidth, % width of caption is 90% of current textwidth
    labelfont=bf,        % the label, e.g. figure 12, is bold
    font=small,          % the whole caption text (label + content) is small
    %format=hang,         % no caption text under the label
}

% Define new commands
\newcommand{\SPZ}{SPZ}
\newcommand{\SPZl}{SPZ }
\newcommand{\SPZs}{SPZs}
\newcommand{\SPZsl}{SPZs }
\newcommand{\ID}{id}
\newcommand\lword[1]{\leavevmode\nobreak\hskip0pt plus\linewidth\penalty50\hskip0pt plus-\linewidth\nobreak\textbf{#1}}
\newcommand{\operator}[1]{{\normalfont \texttt{#1}}}

% Theorem styles
\newtheoremstyle{style}% name of the style to be used
  {\topsep}% measure of space to leave above the theorem. E.g.: 3pt
  {\topsep}% measure of space to leave below the theorem. E.g.: 3pt
  {\itshape}% name of font to use in the body of the theorem
  {0pt}% measure of space to indent
  {\bfseries}% name of head font
  {:}% punctuation between head and body
  { }% space after theorem head; " " = normal interword space
  {\thmname{#1}\thmnumber{ #2}\thmnote{ (#3)}}

% Define environments
\theoremstyle{style}
\newtheorem{definition}{Definition}
\newtheorem{proposition}{Proposition}
\newtheorem{theorem}{Theorem}
\newtheorem{lemma}{Lemma}
\newtheorem{remark}{Remark}
\newtheorem{assumption}{Assumption}
\newtheorem{example}{Example}
\newtheorem{corollary}{Corollary}
\newtheorem{problem}{Problem}

\setlength{\textfloatsep}{0.1cm}

\maketitle

\begin{abstract}
We introduce \textit{sparse polynomial zonotopes}, a new set representation for formal verification of hybrid systems. Sparse polynomial zonotopes can represent non-convex sets and are generalizations of zonotopes, polytopes, and Taylor models. Operations like Minkowski sum, quadratic mapping, and reduction of the representation size can be computed with polynomial complexity w.r.t. the dimension of the system. In particular, for reachability analysis of nonlinear systems, the wrapping effect is substantially reduced using sparse polynomial zonotopes, as demonstrated by numerical examples. In addition, we can significantly reduce the computation time compared to zonotopes when dealing with nonlinear dynamics. 
\end{abstract}

\begin{IEEEkeywords}
Reachability analysis, nonlinear dynamics, hybrid systems, sparse polynomial zonotopes. 
\end{IEEEkeywords}

% INTRODUCTION ---------------------------------------------

\section{Introduction}
\label{sec:Introduction}

\IEEEPARstart{E}{fficient} set representations are highly relevant for controller synthesis and formal verification of hybrid systems, since many underlying algorithms compute with sets; see e.g., \cite{Bravo2006, Zamani2012, Kaynama2012, Schuermann2017c}. Improvements originating from a new set representation often significantly reduce computation time and improve the accuracy of set-based computations.

\subsection{State of the Art}

Fig.~\ref{fig:RelationsSetRepresentations} shows relevant set representations and their relations to each other. Almost all typical set representations are convex, except Taylor models, star sets, level sets, and polynomial zonotopes. Since all convex sets can be represented by support functions, which are closed under Minkowski addition, linear maps, and convex hull operations, they are a good choice for reachability analysis \cite{Girard2008b,Frehse2011,Frehse2015,Schupp2017,Ray2015}. Ellipsoids and polytopes are special cases of support functions, which are often used for reachability analysis \cite{Chutinan2003,Kurzhanskiy2007,Schupp2017,Bogomolov2018,Dang2012} and computations of invariant sets \cite{Blanchini1999,Alamo2005,Rungger2017,Legat2018}. However, the disadvantage of ellipsoids is that they are not closed under intersection and Minkowski addition; the disadvantage of polytopes is that the Minkowski sum is computationally expensive \cite{Tiwary2008}. 

\begin{figure}
\begin{center}
	\includegraphics[width = 0.45 \textwidth]{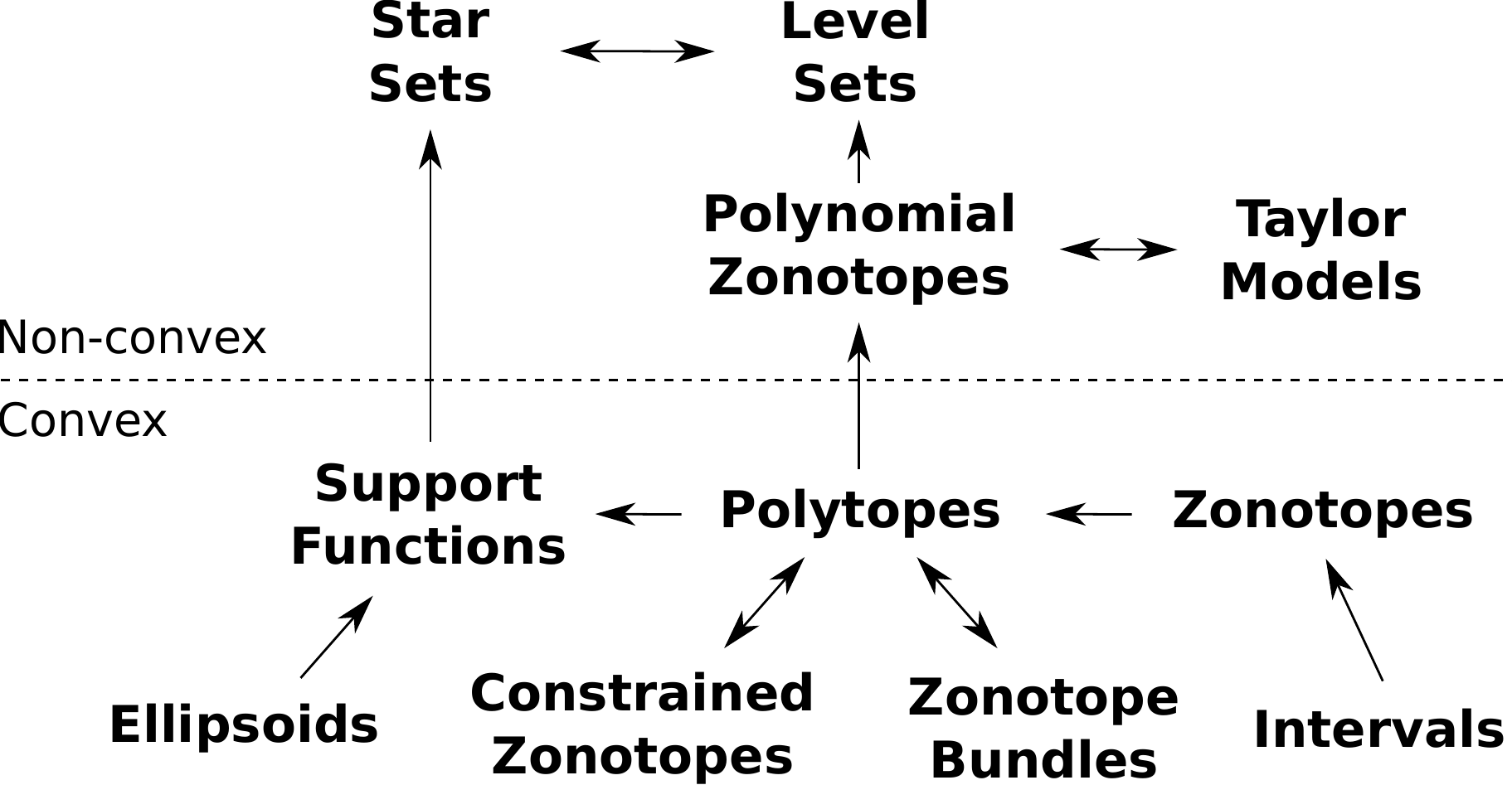}
	\caption{Visualization of the relations between the different set representations, where A $\rightarrow$ B denotes that B is a generalization of A.}
	\label{fig:RelationsSetRepresentations}
	\end{center}
\end{figure}

One important subclass of polytopes is zonotopes, which can be represented compactly by so-called generators: a zonotope with $l$ generators in $n$ dimensions might have up to ${l\choose n-1}$ halfspaces. More importantly, Minkowski sum and linear maps can be computed cheaply, making them a good choice for reachability analysis \cite{Girard2005,Althoff2016c,Immler2015b,Immler2015,Schupp2017}. Zonotopes are closely related to affine arithmetic \cite{deFigueiredo2004} with the zonotope factors being identical to the noise symbols in affine arithmetic. Two relevant extensions to zonotopes are zonotope bundles \cite{Althoff2011f}, where the set is represented implicitly by the intersection of several zonotopes, and constrained zonotopes \cite{Scott2016}, where additional equality constraints on the zonotope factors are considered. Zonotope bundles, as well as constrained zonotopes, are both able to represent any bounded polytope. Both representations make use of lazy computations and thus suffer much less from the curse of dimensionality, as is the case for polytopes \cite{Tiwary2008}. Two other set representations related to zonotopes are complex zonotopes \cite{Adimoolam2016} and zonotopes with sub-polyhedric domains \cite{Ghorbal2010}. Complex zonotopes are defined by complex valued vectors and are well-suited to verify global exponential stability for systems with complex valued eigenvectors \cite{Adimoolam2016}. Zonotopes with sub-polyhedric domains use zones, octagons, and polyhedra instead of intervals to represent the domain for the zonotope factors, which enables the efficient computation of intersections and unions of sets by exploiting lazy computations \cite{Ghorbal2010}. A special case of zonotopes are multi-dimensional intervals, which are particularly useful for range bounding of nonlinear functions via interval arithmetic \cite{Jaulin2006}, but they are also used for reachability analysis \cite{Eggers2012,Ramdani2011b}. Since intervals are not closed under linear maps, one often has to split them to reduce the wrapping effect \cite{Lohner2001}.

In general, reachable sets of nonlinear systems are non-convex, so that tight enclosures can only be achieved using non-convex set representations when avoiding the splitting of reachable sets. Taylor models \cite{Makino2003}, which consist of a polynomial and an interval remainder part, are an example of non-convex set representation. They are typically used for range bounding \cite{Makino2005} and reachability analysis \cite{Chen2012,Chen2013,Makino2009,Neher2007}. Polynomial zonotopes, another type of non-convex set representation, are introduced in \cite{Althoff2013a} and can equally represent the set defined by a Taylor model, as later shown in this work. Quadratic zonotopes \cite{Adje2015} are a special case of polynomial zonotopes. Two other ways to represent non-convex sets are star sets, which are especially useful for simulation-based reachability analysis \cite{Duggirala2016,Bak2017}, and level sets of nonlinear functions \cite{Mitchell2008b}, which are applied to compute reachable sets \cite{Mitchell2005} and controlled invariant regions \cite{Korda2014}. While star sets and level sets are very expressive, it is yet unclear how some operations, such as nonlinear mapping, are computed.

\subsection{Overview}

In this work, we introduce a new non-convex set representation called \textit{sparse polynomial zonotopes}, which is a non-trivial extension of polynomial zonotopes from \cite{Althoff2013a} and exhibits the following major advantages: a) sparse polynomial zonotopes enable a very compact representation of sets, b) they are closed under all relevant operations, c) many other set representations can be converted to a sparse polynomial zonotope, and most importantly, d) all operations have at most polynomial complexity.

The remainder of this paper is structured as follows: In Sec. \ref{sec:SparsePolynomialZonotope}, the formal definition of sparse polynomial zonotopes is provided and important operations on them are derived. We show how sparse polynomial zonotopes provide substantially better results for reachability analysis in Sec. \ref{sec:ReachabilityAnalysis}, which is demonstrated in Sec. \ref{sec:NumericalExamples} on four numerical examples.

\subsection{Notation}
\label{subsec:notation}

In the remainder of this paper, we will use the following notations: Sets are always denoted by calligraphic letters, matrices by uppercase letters, and vectors by lowercase letters. Given a discrete set $\mathcal{H} \in \{ \cdot \}^n$, $|\mathcal{H}| = n$ denotes the cardinality of the set. Given a vector $b \in \mathbb{R}^n$, $b_{(i)}$ refers to the $i$-th entry. Given a matrix $A \in \mathbb{R}^{n \times m}$, $A_{(i,\cdot)}$ represents the $i$-th matrix row, $A_{(\cdot,j)}$ the $j$-th column, and $A_{(i,j)}$ the $j$-th entry of matrix row $i$. Given a discrete set of positive integer indices $\mathcal{H} = \{h_1,\dots,h_{|\mathcal{H}|} \}$ with $1 \leq h_i \leq m ~ \forall i \in \{ 1, \dots, |\mathcal{H}| \}$, $A_{(\cdot,\mathcal{H})}$ is used for $[ A_{( \cdot,h_1 )} ~ \dots ~ A_{( \cdot, h_{|\mathcal{H}|} )} ]$, where $[C~D]$ denotes the concatenation of two matrices $C$ and $D$. The symbols $\mathbf{0}^{(n,m)} \in \mathbb{R}^{n \times m}$ and $\mathbf{1}^{(n,m)} \in \mathbb{R}^{n \times m}$ represent matrices of zeros and ones, and $I_n \in \mathbb{R}^{n \times n}$ is the identity matrix. The empty matrix is denoted by $[~]$. The left multiplication of a matrix $M \in \mathbb{R}^{m \times n}$ with a set $\mathcal{S} \subset \mathbb{R}^n$ is defined as $M \otimes \mathcal{S} = \{ M s ~ | ~ s \in \mathcal{S} \}$, the Minkowski addition of two sets $\mathcal{S}_1 \subset \mathbb{R}^n$ and $\mathcal{S}_2 \subset \mathbb{R}^n$ is defined as $\mathcal{S}_1 \oplus \mathcal{S}_2 = \{ s_1 + s_2 ~|~ s_1 \in \mathcal{S}_1, s_2 \in \mathcal{S}_2 \}$, and the Cartesian product of two sets $\mathcal{S}_1 \subset \mathbb{R}^n$ and $\mathcal{S}_2 \subset \mathbb{R}^m$ is defined as $\mathcal{S}_1 \times \mathcal{S}_2 = \{ [s_1~s_2]^T ~|~ s_1 \in \mathcal{S}_1, s_2 \in \mathcal{S}_2\}$. We further introduce an $n$-dimensional interval as $\mathcal{I} = [l,u],~ \forall i ~ l_{(i)} \leq u_{(i)},~ l,u \in \mathbb{R}^n$. The Nabla operator is defined as $\nabla = \sum_{i=1}^{n} e_i \frac{\partial}{\partial x_{(i)}}$, with $x \in \mathbb{R}^n$ and $e_i \in \mathbb{R}^n$ being orthogonal unit vectors. For the derivation of computational complexity, we consider all binary operations except concatenations, and initializations are explicitly not considered.

% SPARSE POLYNOMIAL ZONOTOPES ------------------------------

\section{Sparse Polynomial Zonotopes} 
\label{sec:SparsePolynomialZonotope}

\begin{figure}
\begin{center}
	\includegraphics[width = 0.4 \textwidth]{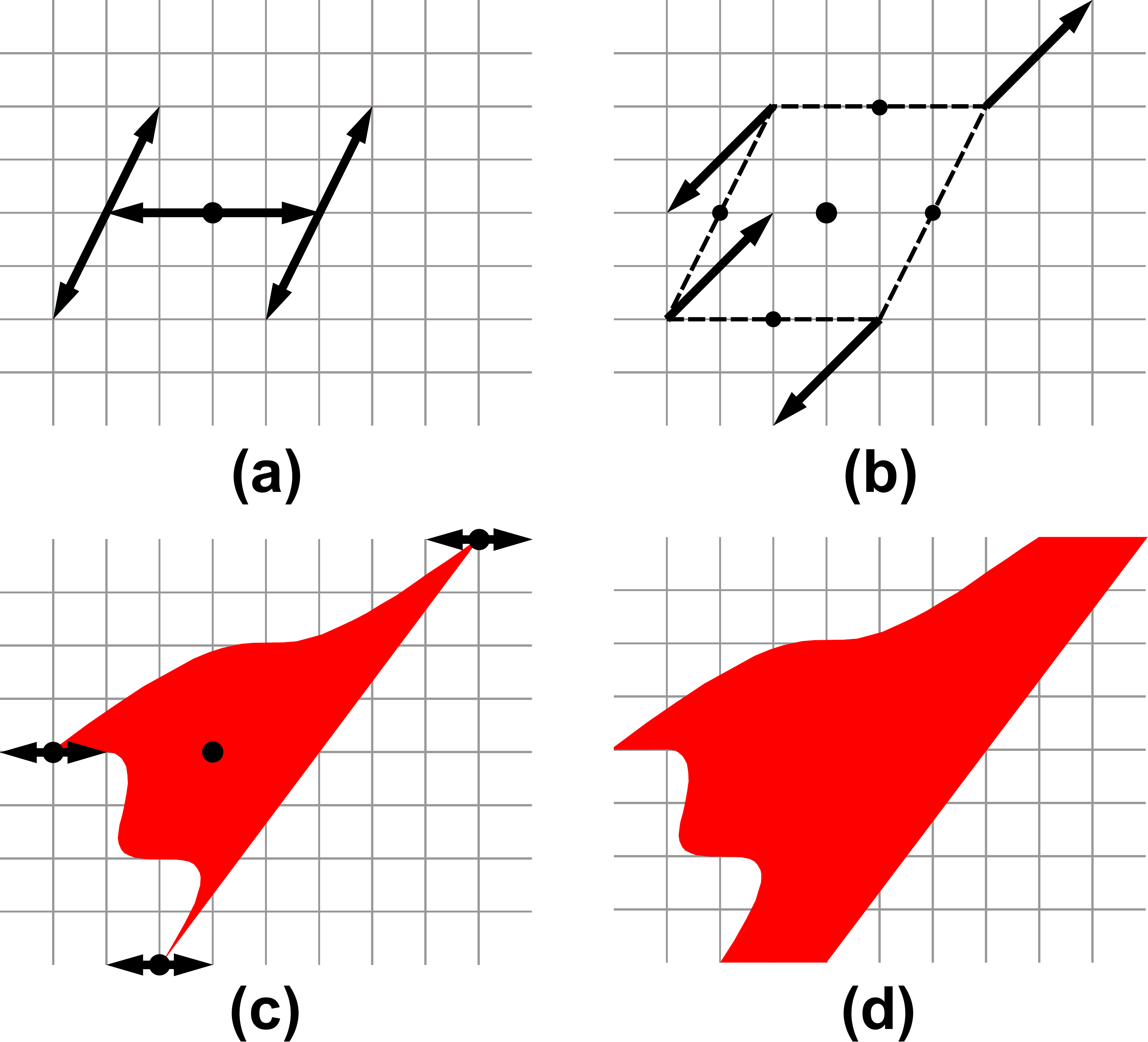}
	\caption{Construction of the \SPZl in Example \ref{ex:PolyZonotope} from the single generator vectors.}
	\label{fig:Example}
	\end{center}
\end{figure}

Let us first define sparse polynomial zonotopes (\SPZs), followed by derivations of relevant operations on them. 
\begin{definition}
  (Sparse Polynomial Zonotope) Given a generator matrix of dependent generators $G \in \mathbb{R}^{n \times h}$, a generator matrix of independent generators $G_I \in \mathbb{R}^{n \times q}$, and an exponent matrix $E \in \mathbb{N}_{0}^{p \times h}$, an \SPZl is defined as  
  \begin{equation}
  	\begin{split}
    \mathcal{PZ} = \bigg\{ & \sum _{i=1}^h \bigg( \prod _{k=1}^p \alpha _k ^{E_{(k,i)}} \bigg) G_{(\cdot,i)} + \sum _{j=1}^{q} \beta _j G_{I(\cdot,j)} ~ \bigg| \\ & \alpha_k, \beta_j \in [-1,1] \bigg\}.
    \end{split}
  \label{eq:polyZonotope}
  \end{equation}
\end{definition}  

% \vskip 10pt plus -1fil

The scalars $\alpha_k$ are called \textit{ dependent factors} since a change in their value affects multiplication with multiple generators. Consequently, the scalars $\beta_j$ are called \textit{independent factors} because they only affect multiplication with one generator. The number of dependent factors is $p$, the number of independent factors is $q$, and the number of dependent generators is $h$. The order of an \SPZl $\rho$ is defined as $\rho = \frac{h + q}{n}$. The independent generators are required for computational reasons: while computations on the dependent generators are exact but computational expensive, computations on the independent generators are often over-approximative but fast. For the derivation of the computational complexity of set operations, we introduce 
\begin{equation}
	h = c_h n, ~ p = c_p n ,~ q = c_q n,
	\label{eq:complexity}
\end{equation} 
with $c_h,c_p,c_q \in \mathbb{R}_{\geq 0}$. The assumption in \eqref{eq:complexity} is justified by the fact that we limit the order $\rho$ of an \SPZl to stay below a desired value $\rho_d$. In the remainder of this paper, we call the term $\alpha_1^{E_{(1,i)}} \cdot \dotsc \cdot \alpha_p^{E_{(p,i)}} \cdot G_{(\cdot,i)}$ a monomial, and $\alpha_1^{E_{(1,i)}} \cdot \dotsc \cdot \alpha_p^{E_{(p,i)}}$ the variable part of the monomial. In order to keep track of the dependencies between the dependent factors from different \SPZs, we introduce an unambiguous integer identifier for each dependent factor $\alpha_k$ and store the identifiers for all dependent factors in a row vector $\ID \in \mathbb{N}_{>0}^{1 \times p}$. Using this identifier vector, we introduce the shorthand $\mathcal{PZ} = \langle G, G_I, E, \ID \rangle_{PZ} \subset \mathbb{R}^n$ for the representation of \SPZs.
All components of a set $\square_i$ have the index $i$, e.g., $p_i$, $h_i$, and $q_i$ belong to $\mathcal{PZ}_i$. Since many set operations require the generation of new unique identifiers we introduce the operation $\operator{uniqueID}(m)$, which returns an identifier vector of length $m$ that contains newly generated unique identifiers. The concept of using unique identifiers to keep track of dependencies is also used in \cite{Combastel2019}, so that the operation \operator{uniqueID} can be implemented as in \cite[Tab.~3]{Combastel2019}. To make \SPZsl more intuitive, we introduce the following example:

\begin{example} 
	The \SPZl
	\begin{equation*}
		\mathcal{PZ} = \left\langle \begin{bmatrix} 4 & 2 & 1 & 2 \\ 4 & 0 & 2 & 2 \end{bmatrix}, \begin{bmatrix} 1 \\ 0 \end{bmatrix}, \begin{bmatrix} 0 & 1 & 0 & 3 \\ 0 & 0 & 1 & 1 \end{bmatrix}, [1 ~ 2 ] \right\rangle_{PZ}
	\end{equation*}
	defines the set
	\begin{equation*}
	\begin{split}
  		\mathcal{PZ} = \bigg\{ & \begin{bmatrix} 4 \\ 4 \end{bmatrix} + \begin{bmatrix} 2 \\ 0 	\end{bmatrix} \alpha_1 + \begin{bmatrix} 1 \\ 2 \end{bmatrix} \alpha_2 + \begin{bmatrix} 2 \\ 2 \end{bmatrix} \alpha_1^3 \alpha_2 + \begin{bmatrix} 1 \\ 0 \end{bmatrix} \beta_1 ~ \bigg| \\ 
  		& ~ \alpha_1, \alpha_2, \beta_1 \in [-1,1] \bigg\}.
	\end{split}
	\end{equation*}
The construction of this \SPZl is visualized in Fig.~\ref{fig:Example}: (a) shows the set spanned by the constant offset vector and the second and third dependent generator, (b) shows the addition of the dependent generator with the mixed term $\alpha_1^3 \alpha_2$, (c) shows the addition of the independent generator, and (d) visualizes the final set. 
	\label{ex:PolyZonotope}
\end{example}

\SPZsl are a more compact representation of polynomial zonotopes \cite{Althoff2013a}, resulting in completely different algorithms for operations on them. In \cite[Def. 1]{Althoff2013a}, the generators $g^{([i],j,k,\dots,m)}$ for all possible combinations of dependent factors up to a certain polynomial degree $\mu$ are stored: 
\begin{equation*}
	\begin{split}
		\mathcal{PZ} = \bigg \{ & c + \sum_{j=1}^p \alpha_j ~ g^{([1],j)} + \sum_{j=1}^p \sum_{k=j}^p \alpha_j \alpha_k ~ g^{([2],j,k)} + \dots + \\
		& \sum_{j=1}^p \sum_{k=1}^p \dots \sum_{m=l}^p \alpha_j \alpha_k \dots \alpha_m ~ g^{([\mu],j,k,\dots,m)} \\
		& + \sum_{i = 1}^q \beta_j ~ G_{I(\cdot,j)} ~\bigg | ~ \alpha_i, \beta_j \in [-1,1] \bigg \},
	\end{split}
\end{equation*}
with $g^{([\mu],j,k,\dots,m)} \in \mathbb{R}^n$, $c \in \mathbb{R}^n$, $G_I \in \mathbb{R}^{n \times q}$. This results in $h = {\mu + p \choose p }$ generators \cite[Chapter~3~(3.8)]{Reimer2003}. For the one-dimensional polynomial zonotope $\mathcal{PZ} = \{ \alpha_1 \cdot \ldots \cdot \alpha_{19} \cdot \alpha_{20}^{10} | \alpha_i \in [-1,1] \}$ with $p=20$ dependent factors and with a polynomial degree of $\mu = 10$, the number of dependent generators is $h = 30045015$. This demonstrates that the number of stored generators can become very large if the polynomial degree and the number of dependent factors are high, which makes computations on the previous set representation very inefficient. Even in comparison with quadratic zonotopes, which correspond to a polynomial order of $\mu = 2$, \SPZsl have lower or equal complexity for all set operations considered in \cite{Althoff2013a} (see Tab.~\ref{tab:complexity}). We in turn use a sparse representation, where only the generators for desired factor combinations are stored, which enables the representation of the above polynomial zonotope with only one single generator. Furthermore, our representation does not require limiting the polynomial degree of the polynomial zonotope in advance, which is advantageous for reachability analysis, as shown in Sec. \ref{subsec:advantages}.  

\begin{table}[h]
\begin{center}
\caption{Computational complexity with respect to the dimension $n$.}
\label{tab:complexity}
\begin{tabular}{l c c}
 \toprule
 \textbf{Set Operation} & \textbf{SPZ} & \textbf{Quad. Zono. \cite{Althoff2013a}}  \\ \midrule 
 Multiplication with matrix & $\mathcal{O}(n^2 m)$ & $\mathcal{O}(n^2 m)$ \\
 Mink. add. with zonotope & $\mathcal{O}(1)$ & $\mathcal{O}(n)$ \\
 Enclosure by zonotope & $\mathcal{O}(n^2)$ & $\mathcal{O}(n^2)$ \\
 Quadratic map & $\mathcal{O}(n \log(n) + n^3 m)$ & $\mathcal{O}(n^4 m)$ \\
 \bottomrule 
\end{tabular}
\end{center}
\end{table}

% Preliminary operations

\subsection{Preliminary Operations}

First, we introduce preliminary set operations that are required for many other operations.

% Merge of set identifiers

\subsubsection{Merging the Set of Identifiers}

For all set operations that involve two or more \SPZs, the operator \operator{mergeID} is necessary in order to build a common representation of exponent matrices to fully exploit the dependencies between identical dependent factors:

\begin{proposition}
	(Merge ID) Given two \SPZs, $\mathcal{PZ}_1 = \langle G_1,\linebreak[1] G_{I,1}, E_1, \ID_1 \rangle_{PZ}$ and $\mathcal{PZ}_2 = \langle G_2, \linebreak[1] G_{I,2},$ $E_2, \ID_2 \rangle_{PZ}$, \operator{mergeID} returns two adjusted \SPZsl with identical identifier vectors that are equivalent to $\mathcal{PZ}_1$ and $\mathcal{PZ}_2$, and has a complexity of $\mathcal{O}(n^2)$:
	\begin{equation*}
	\begin{split}
		& \operator{mergeID}(\mathcal{PZ}_1,\mathcal{PZ}_2) = \big \{ \langle G_1, G_{I,1}, \overline{E}_1, \overline{\ID} \rangle_{PZ}, \\
		& ~~~~~~~~~~~~~~~~~~~~~~~~~~~~~~~~~ \langle G_2, G_{I,2}, \overline{E}_2, \overline{\ID} \rangle_{PZ} \big \} \\
		& ~~ \\
		& \text{with} ~~ \overline{\ID} = \begin{bmatrix} \ID_1 & \ID_{2(\cdot,\mathcal{H})} \end{bmatrix},~~\mathcal{H} = \left\{ i~ |~ \ID_{2(i)} \not\in \ID_1 \right\}, \\
		& ~~~~~~~ \overline{E}_1 = \begin{bmatrix} E_1 \\ \mathbf{0}^{(|\mathcal{H}|,h_1)} \end{bmatrix} \in \mathbb{R}^{a \times h_1}, \\
		& ~~~~~~~ \overline{E}_{2(i,\cdot)} = \begin{cases} E_{2(j,\cdot)},~ \mathrm{if} ~ \exists j~\overline{\ID}_{(i)} = \ID_{2(j)} \\ \mathbf{0}^{(1,h_2)}, ~\mathrm{otherwise} \end{cases} i = 1 \dots a,
	\end{split}
	\end{equation*}
where $a = |\mathcal{H}|+p_1$.
\label{prop:mergeID}
\end{proposition}

\begin{proof}
The extension of the exponent matrices with all-zero rows only changes the representation of the set, but not the set itself.
	
Complexity:	The only operation with super-linear complexity with respect to the system dimension $n$ is the construction of the set $\mathcal{H}$ with $\mathcal{O}(p_1 p_2) = \mathcal{O}(n^2)$ using \eqref{eq:complexity}.
\end{proof}

% Transformation to a minimal representation

\subsubsection{Transformation to a Compressed Representation}

Some set operations result in an \SPZl that contains multiple monomials with an identical variable part, which we combine to one single monomial:

\begin{proposition}
	(Compact) Given an \SPZl $\mathcal{PZ} = \langle G,G_I,\linebreak[0]E,\ID \rangle_{PZ}$, the operation \operator{compact} returns a compressed representation of the set $\mathcal{PZ}$ and has a complexity of $\mathcal{O}(n^2 \log (n))$:
	\begin{equation*}
		\begin{split}
			& \operator{compact}(\mathcal{PZ}) = \langle \overline{G}, G_I, \overline{E}, \ID \rangle_{PZ} \\
			& ~~ \\
			& \text{with} ~~ \overline{E} = \operator{uniqueColumns}( E ) \in \mathbb{N}_{0}^{p \times k}, \\
			& ~~~~~~~~ \mathcal{H}_j = \big \{ i~ |~ \overline{E}_{(l,j)} = E_{(l,i)} ~ \forall l \in \{1, \dots, p\} \big \}, \\
			& ~~~~~~~~ \overline{G} = \bigg[ \sum_{i \in \mathcal{H}_1} G_{(\cdot,i)} ~ \dots ~ \sum_{i \in \mathcal{H}_k} G_{(\cdot,i)} \bigg],
		\end{split}
	\end{equation*}
	where the operation \operator{uniqueColumns} removes identical matrix columns until all columns are unique.
	\label{prop:compact}
\end{proposition}

\begin{proof}
	For an \SPZl where the exponent matrix $E = [e ~ e]$ consists of $2$ identical columns $e \in \mathbb{N}_{0}^{p}$, it holds that 
	\begin{equation*}
	\begin{split}
		& \bigg \{ \bigg( \prod _{k=1}^p \alpha _k ^{e_{(k)}} \bigg) G_{(\cdot,1)} + \bigg( \prod _{k=1}^p \alpha _k ^{e_{(k)}} \bigg) G_{(\cdot,2)} ~\bigg| ~ \alpha_k \in [-1,1] \bigg \} \\
		& = \bigg \{  \bigg( \prod _{k=1}^p \alpha _k ^{e_{(k)}} \bigg) \bigg ( G_{(\cdot,1)} + G_{(\cdot,2)} \bigg) ~\bigg| ~ \alpha_k \in [-1,1] \bigg \}.
	\end{split}
	\end{equation*}
	Summation of the generators for monomials with identical exponents therefore does not change the set, which proves that $\operator{compact}(\mathcal{PZ}) \equiv \mathcal{PZ}$. Furthermore, since the number of unique columns $k$ of matrix $E$ is smaller than the number of overall columns $h$, the matrices $\overline{E}$ and $\overline{G}$ are smaller or of equal size compared to the matrices $E$ and $G$, which results in a compressed representation of the set. 
	
Complexity:	The operation \operator{uniqueColumns} in combination with the construction of the sets $\mathcal{H}_j$ can be efficiently implemented by sorting the matrix columns followed by an identification of identical neighbors, which can be realized with a worst-case complexity of $\mathcal{O}(p h \log (h))$ \cite[Chapter~5]{Knuth1997}. Since all other operations have at most quadratic complexity, the overall complexity is $\mathcal{O}(n^2 \log (n))$ using \eqref{eq:complexity}.
\end{proof}

% Conversion from other set representations

\subsection{Conversion from other Set Representations}

This subsection shows how other set representations can be converted to \SPZs.

% Zonotope

\subsubsection{Zonotope and Interval}

We first provide the definition of a zonotope:

\begin{definition}
	(Zonotope) \cite[Def. 1]{Girard2005} Given a center $c \in \mathbb{R}^n$ and a generator matrix $G \in \mathbb{R}^{n \times l}$, a zonotope is defined as 
	\begin{equation}
		\mathcal{Z} = \bigg\{ c + \sum_{i=1}^l \alpha_i ~ G_{(\cdot,i)} ~\bigg|~ \alpha_i \in [-1,1] \bigg\}.
	\label{eq:zonotope}
	\end{equation}
\end{definition}

For a compact notation, we introduce the shorthand $\mathcal{Z} = \langle c,G \rangle_{Z}$. Any zonotope can be converted to an \SPZ:

\begin{proposition}
	(Conversion Zonotope) A zonotope $\mathcal{Z} = \langle c, G \rangle_{Z}$ can be represented by an \SPZ:
	\begin{equation*}
		\mathcal{Z} = \left\langle [c ~ G] ,[~],[\mathbf{0}^{(n,1)} ~ I_l],\operator{uniqueID}(l) \right\rangle_{PZ}.
	\end{equation*}
	The complexity of the conversion is $\mathcal{O}(l)$, with $l$ denoting the number of zonotope generators.
	\label{prop:conversionZonotope}
\end{proposition}

\begin{proof}
	If we insert $E = [\mathbf{0}^{(n,1)} ~ I_l]$ and $G_I = [~]$ into \eqref{eq:polyZonotope}, we obtain a zonotope (see \eqref{eq:zonotope}). 
	
Complexity:	The construction of the matrices only involves concatenations, and therefore has complexity $\mathcal{O}(1)$. Generation of $l$ unique identifiers has complexity $\mathcal{O}(l)$, resulting in an overall complexity of $\mathcal{O}(1) + \mathcal{O}(l) = \mathcal{O}(l)$.
\end{proof}

Since any interval can be represented as a zonotope \cite[Prop. 2.1]{Althoff2010a}, their conversion to an \SPZl is straightforward.

% Polytope

\subsubsection{Polytope}

We start with the vertex-representation of a polytope:
\begin{definition}
	(Polytope) \cite[Def. 2.2]{Althoff2010a} Given $r$ polytope vertices $v_i \in \mathbb{R}^n,~i \in \{1, \dots, r\}$, a polytope $\mathcal{P}$ is defined as 
	\begin{equation*}
		\mathcal{P} = \bigg\{  \sum_{i=1}^{r} \lambda_i v_i ~ \bigg | ~ \lambda_i \geq 0,~ \sum_{i=1}^r \lambda_i = 1  \bigg \}.
	\end{equation*}
	\label{def:polytope}
\end{definition}

We use the shorthand $\mathcal{P} = \langle [v_1 ~ \dots ~v_r] \rangle_P$.

\begin{theorem}
	(Conversion Polytope) Every bounded polytope can be equivalently represented as an \SPZ. 
\end{theorem}

\begin{proof}
	As shown in Def.~\ref{def:polytope}, every bounded polytope $\mathcal{P}$ can be described as the convex hull of its vertices. Each vertex $v_i \in \mathbb{R}^n$ can be equivalently represented as an \SPZl without independent generators $v_i = \langle v_i, [~], 0,\operator{uniqueID}(1) \rangle_{PZ}$. Since \SPZsl without independent generators are closed under the convex hull operation, as we show later in Prop.~\ref{prop:convHull1}, computing the convex hull of all vertices results in an \SPZl that is equivalent to the bounded polytope $\mathcal{P}$.
\end{proof}

An algorithm for the conversion of a polytope in vertex-representation to an \SPZl is provided in \cite[Alg.~1]{Kochdumper2019a}.

% Taylor model

\subsubsection{Taylor Model}

First, we formally define multi-dimensional Taylor models:

\begin{definition}
	(Taylor Model) \cite[Def. 2.1]{Chen2012} Given a vector field $w: \mathbb{R}^s \to \mathbb{R}^n$, where each sub-function $w_{(i)}: \mathbb{R}^s \to \mathbb{R}$ is a polynomial function defined as  
	\begin{equation}
		w_{(i)} \left( x_1,\dots ,x_s \right) = \sum_{j = 1}^{m_i} b_{i,j} \prod_{k=1}^s	x_{k}^{E_{i(k,j)}},
		\label{eq:polyFuncTaylorModel}	
	\end{equation}
	and an interval $\mathcal{I} \subset \mathbb{R}^n$, a Taylor model $\mathcal{T}(x) \subset \mathbb{R}^n$ is defined as 
	\begin{equation*}
		\mathcal{T}(x) = \left\{ \left. \begin{bmatrix} w_{(1)}(x) \\ \vdots \\ w_{(n)}(x) \end{bmatrix} + \begin{bmatrix} y_{(1)} \\ \vdots \\ y_{(n)} \end{bmatrix} ~ \right| ~ y \in \mathcal{I} \right\},
	\end{equation*}
	where $x = [x_1 ~ \dots ~ x_s]^T$, $E_i \in \mathbb{N}_{0}^{s \times m_i}$ represents an exponent matrix and $b_{i,j} \in \mathbb{R}$ are the polynomial coefficients. 
\end{definition}

For a concise notation, we introduce the shorthand $\mathcal{T}(x) = \langle w(x), \mathcal{I} \rangle_{T}$. The set defined by any Taylor model can be converted to an \SPZ:

\begin{proposition}
	(Conversion Taylor Model) The set defined by a Taylor model $\mathcal{T}(x) = \langle w(x), \mathcal{I} \rangle_{T}$ on the domain $x \in \mathcal{D}$ with $\mathcal{D} = [l_D,u_D]$ and $\mathcal{I} = [l_R,u_R]$ can be equivalently represented by an \SPZ:
	\begin{equation*}
	\mathcal{T}(x) = \left\langle \begin{bmatrix} \frac{l_R + u_R}{2} & \widehat{G} \end{bmatrix} , G_I, \begin{bmatrix} \mathbf{0}^{(s,1)} & \widehat{E} \end{bmatrix}, \operator{uniqueID}(s) \right\rangle_{PZ}
	\end{equation*}
	\begin{equation}
	\begin{split}
			& \text{with} ~~ \widehat{G} = \begin{bmatrix} \big [ \overline{b}_{1,1} ~ \dots ~ \overline{b}_{1,m_1} \big ] &  & \mathbf{0} \\  & \ddots &  \\ \mathbf{0} &  & \big [ \overline{b}_{n,1} ~ \dots ~ \overline{b}_{n,m_n} \big ] \end{bmatrix}, \\
			& ~~~~~~~~ \widehat{E} = \begin{bmatrix} \overline{E}_1 & \dots & \overline{E}_n \end{bmatrix}, \\
			& ~~~~~~~~ G_I = \begin{bmatrix} \frac{u_{R(1)}-l_{R(1)}}{2} & & \mathbf{0} \\ & \ddots & \\ \mathbf{0} & & \frac{u_{R(n)}-l_{R(n)}}{2} \end{bmatrix}.
		\end{split}
		\label{eq:conversionTaylorModel}
	\end{equation}
	The coefficients $\overline{b}_{i,j}$ and the matrices $\overline{E}_i$ result from the definition
	\begin{equation}
		w_{(i)} \left( \delta_1(\alpha_1),\dots ,\delta_s(\alpha_s) \right) := \sum_{j = 1}^{\overline{m}_i} \overline{b}_{i,j} \prod_{k=1}^s	\alpha_k^{\overline{E}_{i(k,j)}},
	\label{eq:TaylorModelTransform}
	\end{equation}
	where $w_{(i)}(\cdot)$ is defined as in \eqref{eq:polyFuncTaylorModel} and
	\begin{equation}	
		\begin{split}
			\delta_k(\alpha_k) = & \frac{l_{D(k)} + u_{D(k)}}{2} + \frac{u_{D(k)} - l_{D(k)}}{2} ~ \alpha_k , \\     
			& \alpha_k \in [-1,1], ~k=1\dots s.
		\end{split}
		\label{eq:auxVar}
	\end{equation}
	The \operator{compact} operation is applied to remove monomials with an identical variable part. The complexity of the conversion is $\mathcal{O}(n^{4+n} \log(n))$.
\end{proposition}

\begin{proof}
	The auxiliary variables $\delta_k(\alpha_k)$ (see \eqref{eq:auxVar}) represent the domain $\mathcal{D}$ with dependent factors $\alpha_k \in [-1,1]$: 
\begin{equation*}
	\mathcal{D} = \big \{ \begin{bmatrix} \delta_1(\alpha_1) & \dots & \delta_s(\alpha_s)\end{bmatrix}^T ~|~ \alpha_1, \dots, \alpha_s \in [-1,1]  \big \}.
\end{equation*}	
Furthermore, the interval $\mathcal{I}= [l_R,u_R]$ can be equivalently represented as a zonotope $\mathcal{I} = \langle 0.5(l_R+u_R), 0.5 ~ \operator{diag}(u_{R(1)}-l_{R(1)}, \dots, u_{R(n)}-l_{R(n)}  ) \rangle_{Z}$ (see \cite[Prop. 2.1]{Althoff2010a}), where the operator \operator{diag} returns a diagonal matrix. The set defined by the Taylor model $\mathcal{T}(x)$, $x \in \mathcal{D}$ can therefore be equivalently expressed as 	
\begin{equation*}
	\begin{split}
		& \mathcal{T}(\delta(\alpha)) = \left\{ \left. \begin{bmatrix} w_{(1)}(\delta(\alpha)) \\ \vdots \\ w_{(n)}(\delta(\alpha)) \end{bmatrix} + \begin{bmatrix} y_{(1)} \\ \vdots \\ y_{(n)} \end{bmatrix} ~ \right| ~ y \in \mathcal{I} \right\} \overset{\substack{ \scriptscriptstyle \eqref{eq:TaylorModelTransform} \\ \scriptscriptstyle \text{and} \\  \scriptscriptstyle \mathcal{I} = [l_{\scaleto{R\mathstrut}{3pt}},u_{\scaleto{R\mathstrut}{3pt}}] }}{=}\\
		& ~ \\
		&  \Bigg\{ \sum_{j = 1}^{\overline{m}_1} \begin{bmatrix} \overline{b}_{1,j} \\ o \end{bmatrix} \prod_{k=1}^s	\alpha_k^{\overline{E}_{1(k,j)}} + \dots + \sum_{j = 1}^{\overline{m}_n} \begin{bmatrix} o \\ \overline{b}_{n,j} \end{bmatrix} \prod_{k=1}^s	\alpha_k^{\overline{E}_{n(k,j)}} \\
		& ~~ + \frac{1}{2} \begin{bmatrix} u_{R(1)} - l_{R(1)} \\ o \end{bmatrix} \beta_1 + \dots + \frac{1}{2} \begin{bmatrix} o \\ u_{R(n)} - l_{R(n)} \end{bmatrix} \beta_n \\
		& ~~ + \frac{l_R + u_R}{2} ~ \bigg | ~ \alpha_k, \beta_1, \dots,\beta_n \in [-1,1]  \Bigg\} \\
		& ~ \\
		& = \left\langle \begin{bmatrix} \frac{l_R + u_R}{2} & \widehat{G} \end{bmatrix} , G_I, \begin{bmatrix} \mathbf{0}^{(s,1)}  & \widehat{E} \end{bmatrix}, \operator{uniqueID}(s) \right\rangle_{PZ},
	\end{split}
\end{equation*}	
where $\delta(\alpha) = [ \delta_1(\alpha_1)~ \dots ~ \delta_s(\alpha_s)]^T$ and $o = \mathbf{0}^{(n-1,1)}$.
	
Complexity: Let $m = \max (m_1,\dots,m_n)$ and $\epsilon = \max ( \linebreak[1] \max(E_1), \dots , \max(E_n))$, where $\max (A)$ returns the maximum entry of a matrix $A$. Since $\delta_k(\alpha_k) = c_0 + c_1 \alpha_k$, $c_0,c_1 \in \mathbb{R}$ (see \eqref{eq:auxVar}), it holds that $\delta_k(\alpha_k)^\epsilon$ is a polynomial in $\alpha_k$ with $\epsilon+1$ polynomial terms. Naive evaluation without intermediate simplification of the function $w_{(i)} \left( \delta_1(\alpha_1),\dots ,\delta_s(\alpha_s) \right)$ with $w_{(i)}(\cdot)$ as defined in \eqref{eq:polyFuncTaylorModel} therefore results in a multivariate polynomial with $\overline{m} = m (\epsilon + 1)^s$ terms in the worst-case. From \eqref{eq:conversionTaylorModel} and \eqref{eq:TaylorModelTransform} it can be deduced that the exponent matrix $\widehat{E}$ consequently consists of at most $\widehat{h} = n \overline{m} = n m (\epsilon + 1)^s$ columns. Since the complexity of \operator{compact} is $\mathcal{O}(\widehat{p} \widehat{h} \log(\widehat{h}))$ (see Prop.~\ref{prop:compact}) and $\widehat{p} = s$, the subsequent application of the \operator{compact} operation has complexity $\mathcal{O}( s n m (\epsilon + 1)^s \log(n m (\epsilon + 1)^s)) = \mathcal{O}(s n m (\epsilon + 1)^s ( \log(n) + \log(m) + s \log(\epsilon + 1))) = \mathcal{O}(n^{4+n} \log(n))$ using \eqref{eq:complexity} and the fact that $s = c_s n$, $m = c_m n$, and $\epsilon = c_\epsilon n$ with $c_s, c_m, c_\epsilon \in \mathbb{R}_{\geq 0}$ holds. This is also the overall complexity of the conversion, since all other operations have a lower complexity. 
\end{proof}

% Over-Approximation

\subsection{Enclosure by other Set Representations}

For computational reasons many algorithms that compute with sets require to enclose sets by simpler set representations. In this subsection we therefore describe how \SPZsl can be enclosed by other set representations. We consider the \SPZ

\begin{equation}
	\begin{split}
	\mathcal{PZ} = \bigg \langle & \begin{bmatrix} -0.5 & 1 & 0 & -1 & 1 \\ -0.5 & 1 & 1 & 1 & 1 \end{bmatrix},[~], \\
	& \begin{bmatrix} 0 & 1 & 0 & 1 & 2 \\ 0 & 0 & 1 & 1 & 0 \end{bmatrix}, [1~2]  \bigg \rangle_{PZ}
	\end{split}
	\label{eq:polyZonotopeEnclose}
\end{equation}

as a running example to demonstrate the tightness of the enclosure.

\begin{figure}
\begin{center}
	\includegraphics[width = 0.49 \textwidth]{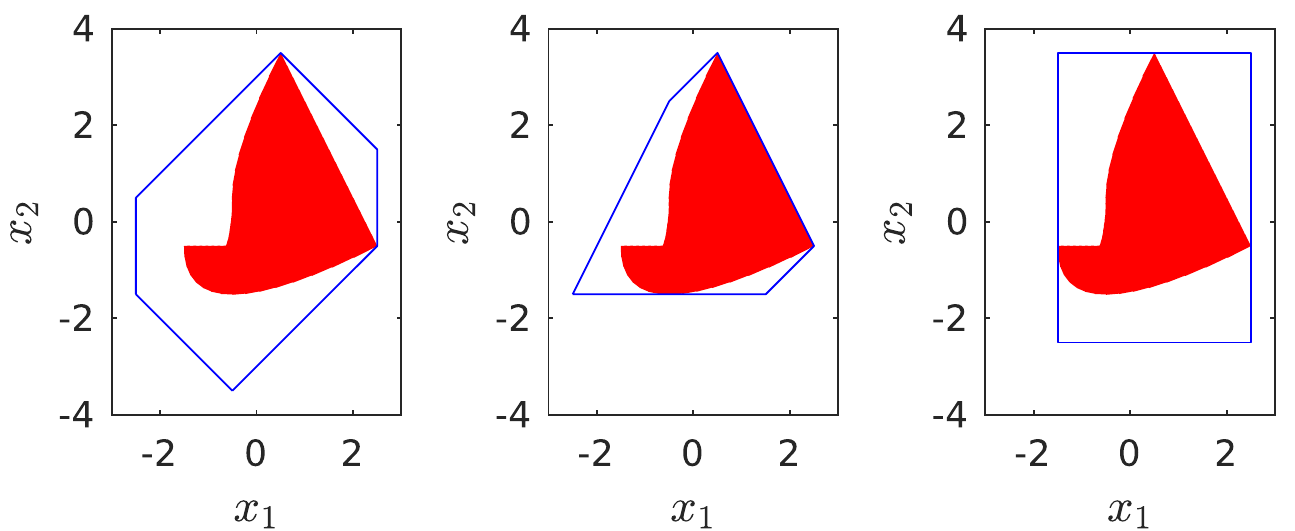}
	\caption{Enclosure of the \SPZl in \eqref{eq:polyZonotopeEnclose} with a zonotope (left), a polytope (middle) and an interval (right).}
	\label{fig:enclosure}
	\end{center}
\end{figure}

% Zonotope

\subsubsection{Zonotope}

We first show how an \SPZl can be enclosed by a zonotope:

\begin{proposition}
	(Zonotope) Given an \SPZl $\mathcal{PZ} = \langle G,G_I,\linebreak[1] E,\ID \rangle_{PZ}$, the operation \operator{zono} returns a zonotope that encloses $\mathcal{PZ}$ with complexity $\mathcal{O}(n^2)$:
	\begin{equation*}
		\begin{split}
			& \operator{zono} \left( \mathcal{PZ} \right) = \bigg\langle \sum_{i \in \mathcal{N}} G_{\left(\cdot, i \right) } + 0.5 \sum_{i \in \mathcal{H}} G_{\left(\cdot, i \right)}, \\
			& ~~~~~~~~~~~~~~~~~~~~ \begin{bmatrix} 0.5 ~ G_{(\cdot,\mathcal{H})} & G_{(\cdot,\mathcal{K})} & G_I\end{bmatrix} \bigg\rangle_{Z} \\
			& ~~ \\
			& \text{with} ~~ \mathcal{N} = \left\{ i ~ \left|~ E_{(j,i)} = 0 ~ \forall j \in \{ 1, \dots, p \} \right. \right\}, \\
			& ~~~~~~~~ \mathcal{H} = \bigg\{ i ~\bigg|~ \prod_{j=1}^p \left(1-(E_{(j,i)} ~\mathrm{mod} ~ 2)\right) = 1 \bigg\} \setminus \mathcal{N}, \\
			& ~~~~~~~~ \mathcal{K} = \{1, \dots ,h\} \setminus ( \mathcal{H} \cup \mathcal{N} ), \\
			& ~~
		\end{split}
	\end{equation*}
where $x ~ \mathrm{mod} ~ y,~ x,y \in \mathbb{R}$ is the modulo operation.
\label{prop:zonoEnclose}
\end{proposition}

\begin{proof}
	We over-approximate all monomial variable parts in \eqref{eq:polyZonotope} with additional independent factors, which yields a zonotope \eqref{eq:zonotope}. Since monomials with exclusively even exponents ($i \in \mathcal{H}$) are strictly positive, we can enclose them tighter using
	\begin{equation*}
		\begin{split}
			\forall i \in \mathcal{H}: ~~ & \left( \prod_{k=1}^p [-1,1]^{E_{(k,i)}} \right) G_{(\cdot, i)} = [0,1] G_{(\cdot, i)} = \\ & ~~~~~~ = 0.5 ~ G_{(\cdot, i)} + [-1,1]~0.5~ G_{(\cdot, i)}.
		\end{split}
	\end{equation*}
	For all other monomials ($i \in \mathcal{K} $), evaluation of the monomial variable part directly results in the interval $[-1,1]$. A dependent factor affects all monomials that contain the dependent factor. Since the over-approximation of the monomial variable parts with new independent factors destroys this dependence between different monomials (e.g., $\{\alpha_1 \alpha_2^2 + \alpha_1^3 \alpha_2 | \alpha_1, \alpha_2 \in [-1,1] \} \subseteq \{ \beta_1 + \beta_2 | \beta_1,\beta_2 \in [-1,1] \}$), the resulting zonotope encloses $\mathcal{PZ}$, because removing dependence results in an over-approximation \cite{Jaulin2006}.
	
Complexity:	The calculation of the set $\mathcal{H}$ has complexity $\mathcal{O}(ph)$, and the construction of the zonotope is $\mathcal{O}(nh)$ in the worst-case where all exponents are exclusively even, resulting in an overall complexity of $\mathcal{O}(ph) + \mathcal{O}(nh)$, which is $\mathcal{O}(n^2)$ using \eqref{eq:complexity}. 
\end{proof}

The enclosing zonotope for the \SPZl in \eqref{eq:polyZonotopeEnclose} calculated according to Prop.~\ref{prop:zonoEnclose} is visualized in Fig.~\ref{fig:enclosure} (left).

% Polytope

\subsubsection{Polytope}

Next, we show how to enclose an \SPZl with a polytope:

\begin{proposition}
	(Polytope) Given an \SPZl $\mathcal{PZ} = \langle G,G_I,E,\ID \rangle_{PZ}$, the operation \operator{poly} returns a polytope that encloses $\mathcal{PZ}$ with complexity $\mathcal{O}(2^{n^2})$:
	\begin{equation*}
		\operator{poly}(\mathcal{PZ}) = \langle [v_1 ~ \dots ~v_r] \rangle_P,
	\end{equation*}
	where the vertex-representation of the polytope $\langle [v_1 ~ \dots ~v_r] \rangle_P$ is computed by applying \cite[Alg.~2]{Kochdumper2019a} to the \SPZ
	\begin{equation}
		\begin{split}
			& \overline{\mathcal{PZ}} = \big \langle [c_z~\overline{G}_{(\cdot,\mathcal{K})}],[G_I~G_z],[\mathbf{0}^{(n,1)}  ~ \overline{E}_{(\cdot,\mathcal{K})}],id \big \rangle_{PZ} \\
			& ~ \\
			& \text{with}~ \mathcal{H} = \big \{ i ~\big |~ \exists j \in \{ 1,\dots,p \}: ~ E_{(j,i)}  > 1 \big \}, \\
			& ~~~~~~ \mathcal{K} = \{1,\dots,h \} \setminus \mathcal{H} \\
			& ~~~~~~ \langle c_z,G_z \rangle_Z = \operator{zono} \big( \langle G_{(\cdot,\mathcal{H})},[~],E_{(\cdot,\mathcal{H})}, id \rangle_{PZ} \big)
		\end{split}
		\label{eq:polyEnclose2}
	\end{equation}
	\label{prop:polyEnclose}
\end{proposition}

\begin{proof}
Alg.~2 in \cite{Kochdumper2019a} computes an enclosing polytope for an \SPZl for which the exponent matrix has only zeros or ones as entries. We therefore first split $\mathcal{PZ}$ into one part $\langle G_{(\cdot,\mathcal{K})},G_I,E_{(\cdot,\mathcal{K})},id \rangle_{PZ}$ with only zeros or ones in the exponent matrix, and one remainder part $\langle G_{(\cdot,\mathcal{H})},[~],E_{(\cdot,\mathcal{H})},id \rangle_{PZ}$. In order to remove exponents that are greater than one the remainder part is enclosed by a zonotope (see \eqref{eq:polyEnclose2}) using Prop.~\ref{prop:zonoEnclose}. Combination of the two parts finally yields the \SPZl $\overline{\mathcal{PZ}}$ which satisfies
	\begin{equation*}
		\mathcal{PZ} \subseteq \overline{\mathcal{PZ}} \overset{\text{\cite[Alg.~2]{Kochdumper2019a}}}{\subseteq} \langle [v_1 ~ \dots ~v_r] \rangle_P.
	\end{equation*}
	
Complexity:	The calculation of the sets $\mathcal{H}$ and $\mathcal{K}$ in \eqref{eq:polyEnclose2} has complexity $\mathcal{O}(ph)$, and the computation of an enclosing zonotope has according to Prop.~\ref{prop:zonoEnclose} complexity $\mathcal{O}(n^2)$. According to \cite[Prop.~7]{Kochdumper2019a} the complexity of \cite[Alg.~2]{Kochdumper2019a} is $\mathcal{O}((2^p)^{\lfloor n/2 \rfloor + 1} + 4^p (p + n))$. Using \eqref{eq:complexity}, the overall complexity is therefore $\mathcal{O}(2^{n^2})$.
\end{proof}

The enclosing polytope for the \SPZl in \eqref{eq:polyZonotopeEnclose} calculated according to Prop.~\ref{prop:polyEnclose} is visualized in Fig.~\ref{fig:enclosure} (middle).

% Support Function, Interval and Template Polyhedra

\subsubsection{Support Function, Interval, and Template Polyhedra}
\label{subsubsec:supportFunc}

Let us first derive the support function of an \SPZ.

\begin{definition}
	(Support Function) \cite[Def. 1]{Girard2008b} Given a set $\mathcal{S} \subset \mathbb{R}^n$ and a direction $d \in \mathbb{R}^n$, the support function $s_{\mathcal{S}}: \mathbb{R}^n \to \mathbb{R}$ of $\mathcal{S}$ is defined as 
	\begin{equation*}
		s_{\mathcal{S}}(d) = \max\limits_{x \in \mathcal{S}} ~ d^T x.
	\end{equation*} 
\end{definition}

If $\mathcal{S}$ is convex, its support function is an exact representation; otherwise, an over-approximation is returned. Since \SPZsl are non-convex in general, one can only over-approxi\-mate them by support functions. To compute support functions for \SPZs, we need to introduce the range bounding operation. Given a function $f: \mathbb{R}^{m} \to \mathbb{R}$ and an interval $\mathcal{I} \subset \mathbb{R}^m$, the range bounding operation
	\begin{equation*}
		\operator{B}(f(x),\mathcal{I}) \supseteq \left[\min\limits_{x \in \mathcal{I}} f(x),~ \max\limits_{x \in \mathcal{I}} f(x)\right]
	\end{equation*}
returns an over-approximation of the exact bounds. 

\begin{proposition}
	(Support Function) An \SPZl $\mathcal{PZ} = \langle G,G_I,E,\ID \rangle_{PZ}$ is over-approximated by the support function
	\begin{equation*}
		\begin{split}
			& \widehat{s}_{\mathcal{PZ}}(d) = u + \sum_{j=1}^q \left| \overline{g}_{I(j)} \right| \\
			& ~ \\
			& \text{with} ~~ \langle \overline{g},\overline{g}_I, E, \ID \rangle_{PZ} := d^T \otimes \mathcal{PZ}, \\
			& ~~~~~~~~ [l,u] = \operator{B} \big( w(\alpha_1, \dots , \alpha_{p}),[-\mathbf{1}^{(n,1)},\mathbf{1}^{(n,1)}]\big), \\
			& ~~~~~~~~ w(\alpha_1, \dots , \alpha_{p}) = \sum _{i=1}^{h} \left( \prod _{k=1}^{p} \alpha _k ^{E_{(k,i)}} \right) \overline{g}_{(i)},
		\end{split}
	\end{equation*}
where the vector $d \in \mathbb{R}^n$ specifies the direction. The calculation of $\widehat{s}_{\mathcal{PZ}}(d)$ has complexity $\mathcal{O}(n^2) + \mathcal{O}(\operator{B})$, where $\mathcal{O}(\operator{B})$ denotes the computational complexity of the range bounding operation.
\label{prop:bound}
\end{proposition}  

\begin{proof}
	We first project the \SPZl onto the direction $d$, and then divide the one-dimensional projected \SPZl into one part with dependent generators and one with independent generators: 
\begin{equation*}
	\begin{split}
		d^T \otimes \mathcal{PZ} = & \underbrace{\bigg \{ \sum_{i=1}^h \bigg( \prod _{k=1}^p \alpha _k ^{E_{(k,i)}} \bigg) \overline{g}_{(i)} ~\bigg| ~ \alpha_k \in [-1,1] \bigg \} }_{\subseteq [l,u] ~ \text{(dependent}~\text{part)}} \\
		& \oplus \underbrace{ \bigg \{ \sum_{j=1}^q \beta_j \overline{g}_{I(j)}  ~\bigg | ~ \beta_j \in [-1,1] \bigg \}}_{\overset{\equiv \big [ -\sum_{j=1}^q \left| \overline{g}_{I(j)} \right|,~ \sum_{j=1}^q \left| \overline{g}_{I(j)} \right| \big ]}{\overset{~}{\text{(independent}~\text{part)}}}}.
	\end{split}
\end{equation*}	
The bounds for the independent part calculated by the sum of absolute values are exact \cite[Sec. 2]{Girard2008b}. However, the lower bound $l$ and the upper bound $u$ of the dependent part are over-approximative since the range bounding operation \operator{B} returns an over-approximation, so that $\widehat{s}_{\mathcal{PZ}}(d)$ over-approximates $\mathcal{PZ}$. 

Complexity: The calculation of the projection onto $d$ has a complexity of $\mathcal{O}(nh) + \mathcal{O}(nq)$ (see Prop.~\ref{prop:multiplication}), which results in an overall complexity of $\mathcal{O}(n^2) + \mathcal{O}(\operator{B})$ using \eqref{eq:complexity} since all other operations have linear complexity. 
\end{proof}

Note that the tightness of $\widehat{s}_{\mathcal{PZ}}(d)$ solely depends on the tightness of the bounds of the function $w(\cdot)$ obtained by one of the range bounding techniques, e.g., interval arithmetic \cite{Jaulin2006}, Bernstein polynomials \cite{Cargo1966}, and verified global optimization \cite{Makino2005}. A comparison of different techniques can be found in \cite{Althoff2018}. 

A template polyhedron enclosing an \SPZl can easily be constructed by evaluating the support function $\widehat{s}_\mathcal{PZ}(d)$ for a discrete set of directions $\mathcal{D} = \{d_1, \dots , d_r \},~ d_i \in \mathbb{R}^n,~i = 1 \dots r$. The over-approximation with an interval represents a special case where $\mathcal{D} = \{ I_{n(\cdot,1)}, \linebreak[1] \dots, \linebreak[1] I_{n(\cdot,n)},\linebreak[1] -I_{n(\cdot,1)},\linebreak[1] \dots, \linebreak[1] -I_{n(\cdot,n)} \}$. The enclosing interval for the \SPZl in \eqref{eq:polyZonotopeEnclose} calculated by using Bernstein polynomials for range bounding is visualized in Fig.~\ref{fig:enclosure} (right).

% Basic set operations

\subsection{Basic Set Operations}
\label{sec:basicSetOperations}

This subsection derives basic operations on \SPZs.

% Multiplication with matrix

\subsubsection{Multiplication with a Matrix}
\label{subsubsec:MulMatrix}

The left-multiplication with a matrix is obtained as:

\begin{proposition}
	(Multiplication) Given an \SPZl $\mathcal{PZ} = \langle G, \linebreak[1] G_I,E,\ID \rangle_{PZ} \subset \mathbb{R}^n$ and a matrix $M \in \mathbb{R}^{m \times n}$, the left-multiplication is computed as 
	\begin{equation*}
		M \otimes \mathcal{PZ} = \left\langle MG, MG_I, E, \ID \right\rangle_{PZ},
	\end{equation*}
	which has complexity $\mathcal{O}(mn^2)$.
	\label{prop:multiplication}
\end{proposition}

\begin{proof}
	The result follows directly from inserting the definition of \SPZsl in \eqref{eq:polyZonotope} into the definition of the operator $\otimes$ (see Sec. \ref{subsec:notation}). 
	
Complexity: The complexity results from the complexity of matrix multiplications and is therefore $\mathcal{O}(mnh) + \mathcal{O}(mnq) = \mathcal{O}(mn^2)$ using \eqref{eq:complexity}.
\end{proof}

% Minkowski Sum

\subsubsection{Minkowski Addition}

Even though every zonotope can be represented as an \SPZ, we provide a separate definition for the Minkowski addition of an \SPZl and a zonotope for computational reasons.
\begin{proposition}
	(Addition) Given two \SPZs, $\mathcal{PZ}_1 = \langle G_1, \linebreak[1] G_{I,1}, E_1, \ID_1 \rangle_{PZ}$ and $\mathcal{PZ}_2 = \langle G_2, G_{I,2},$ $E_2, \ID_2 \rangle_{PZ}$, as well as a zonotope $\mathcal{Z} = \langle c_z,G_z \rangle_{Z}$, their Minkowski sum is defined as
	\begin{equation}
		\begin{split}
			\mathcal{PZ}_1 & \oplus \mathcal{PZ}_2 = \\
			 \big\langle & \begin{bmatrix} G_1 & G_2 \end{bmatrix}, \begin{bmatrix} G_{I,1} & G_{I,2} \end{bmatrix}, \\ & \begin{bmatrix} E_1 & \mathbf{0}^{(p_1,h_2)}  \\ \mathbf{0}^{(p_2,h_1)}  & E_2 \end{bmatrix} , \operator{uniqueID}(p_1 + p_2) \big\rangle_{PZ},
		\end{split}
		\label{eq:addition}
	\end{equation}	
	\begin{equation}
		\begin{split}
			\mathcal{PZ}_1 \oplus \mathcal{Z} = \big\langle & \begin{bmatrix} c_z & G_1 \end{bmatrix}, \begin{bmatrix} G_{I,1} & G_z \end{bmatrix}, \\
			 & \begin{bmatrix} \mathbf{0}^{(n,1)} & E_1 \end{bmatrix}, \ID_1 \big\rangle_{PZ},~~~~~~~~~~~~~~~~~
		\end{split}
		\label{eq:additionZonotope}
	\end{equation}
	where \eqref{eq:addition} has complexity $\mathcal{O}(n)$ and \eqref{eq:additionZonotope} has complexity $\mathcal{O}(1)$.  
	\label{prop:addition}
\end{proposition}
\begin{proof}
	The result is obtained by inserting the definition of zonotopes \eqref{eq:zonotope} and \SPZsl \eqref{eq:polyZonotope} into the definition of the Minkowski sum (see Sec. \ref{subsec:notation}). For two \SPZs, we generate new identifiers for all factors since the Minkowski sum per definition removes all dependencies between the two added sets.

Complexity: The construction of the resulting \SPZsl only involves concatenations and therefore has complexity $\mathcal{O}(1)$. For two \SPZs, $p_1 + p_2$ unique identifiers have to be generated, resulting in the complexity $\mathcal{O}(p_1 + p_2)$, which is $\mathcal{O}(n)$ using \eqref{eq:complexity}.
\end{proof}

% Exact Addition

\subsubsection{Exact Addition}

By computation of the Minkowski sum of two \SPZs, $\mathcal{PZ}_1$ and $\mathcal{PZ}_2$, as defined in \eqref{eq:addition}, possible dependencies between $\mathcal{PZ}_1$ and $\mathcal{PZ}_2$ due to common dependent factors get lost. We therefore introduce the exact addition $\mathcal{PZ}_1 \boxplus \mathcal{PZ}_2$ of two \SPZs, which explicitly considers the dependencies between $\mathcal{PZ}_1$ and $\mathcal{PZ}_2$. To bring the exponent matrices to a common representation, we apply \operator{mergeID} prior to the computation.
\begin{proposition}
	(Exact Addition) Given two \SPZs, $\mathcal{PZ}_1 = \langle G_1, \linebreak[1] G_{I,1}, E_1, \ID \rangle_{PZ}$ and $\mathcal{PZ}_2 = \langle G_2, G_{I,2},$ $E_2, \ID \rangle_{PZ}$ with a common identifier vector $\ID$, their exact addition is defined as
	\begin{equation*}
		\begin{split}
			\mathcal{PZ}_1 \boxplus \mathcal{PZ}_2 = \big\langle & \begin{bmatrix} G_1 & G_2 \end{bmatrix}, \begin{bmatrix} G_{I,1} & G_{I,2} \end{bmatrix}, \\ & \begin{bmatrix} E_1 & E_2 \end{bmatrix} , \ID \big\rangle_{PZ},
		\end{split}
	\end{equation*}
	which has complexity $\mathcal{O}(n^2 \log(n))$. The \operator{compact} operation is applied to remove monomials with an identical variable part.	
	\label{prop:exactAddition}
\end{proposition}
\begin{proof}
	The result is identical to the one for Minkowski addition of two \SPZsl in \eqref{eq:addition}, with the difference that the common identifier vector $\ID$ is used instead of newly generated unique identifiers.
	
	Complexity: Merging the identifier vectors has complexity $\mathcal{O}(n^2)$ (see Prop.~\ref{prop:mergeID}). The construction of the resulting \SPZl only involves concatenations and therefore has complexity $\mathcal{O}(1)$. Subsequent application of the \operator{compact} operation has complexity $\mathcal{O}(p_1(h_1 + h_2) \log(h_1 + h_2))$ (see Prop.~\ref{prop:compact}), which is $\mathcal{O}(n^2 \log(n))$ using \eqref{eq:complexity}. The overall complexity is therefore $\mathcal{O}(n^2) + \mathcal{O}(1) + \mathcal{O}(n^2 \log(n)) = \mathcal{O}(n^2 \log(n))$.
\end{proof}

As demonstrated later in Sec. \ref{sec:ReachabilityAnalysis}, using exact addition instead of the Minkowski sum for reachability analysis results in a significantly tighter enclosure of the reachable set.

% Cartesian Product

\subsubsection{Cartesian Product}

Even though every zonotope can be represented as an \SPZ, we provide a separate definition for the Cartesian product of an \SPZl and a zonotope for computational reasons.
\begin{proposition}
	(Cartesian Product) Given two \SPZs, $\mathcal{PZ}_1 = \langle G_1, \linebreak[1] G_{I,1}, E_1, \ID_1 \rangle_{PZ} \subset \mathbb{R}^n$ and $\mathcal{PZ}_2 = \langle G_2, G_{I,2},$ $E_2, \ID_2 \rangle_{PZ}\subset \mathbb{R}^m$, as well as a zonotope $\mathcal{Z} = \langle c_z,G_z \rangle_{Z} \subset \mathbb{R}^m$, their Cartesian product is defined as
	\begin{equation}
		\begin{split}
			\mathcal{PZ}_1 & \times \mathcal{PZ}_2 = \\
			\bigg\langle & \begin{bmatrix} G_1 & \mathbf{0}^{(n,h_2)} \\ \mathbf{0}^{(m,h_1)} & G_2 \end{bmatrix}, \begin{bmatrix} G_{I,1} & \mathbf{0}^{(n,q_2)} \\ \mathbf{0}^{(m,q_1)} & G_{I,2} \end{bmatrix}, \\ & \begin{bmatrix} E_1 & \mathbf{0}^{(p_1,h_2)} \\ \mathbf{0}^{(p_2,h_1)} & E_2 \end{bmatrix} , \operator{uniqueID}(p_1 + p_2) \bigg\rangle_{PZ},
		\end{split}
		\label{eq:cartProd}
	\end{equation}	
	\begin{equation}
		\begin{split}
			\mathcal{PZ}_1 \times \mathcal{Z} = \bigg\langle & \begin{bmatrix} \mathbf{0}^{(n,1)} & G_1 \\ c_z &\mathbf{0}^{(m,h_1)} \end{bmatrix}, \begin{bmatrix} G_{I,1} & \mathbf{0}^{(n,l)} \\ \mathbf{0}^{(m,q_1)} & G_z \end{bmatrix}, \\
			 & \begin{bmatrix} \mathbf{0}^{(p_1,1)} & E_1 \end{bmatrix}, \ID_1 \bigg\rangle_{PZ},
		\end{split}
		\label{eq:cartProdZonotope}
	\end{equation}
	where \eqref{eq:cartProd} has complexity $\mathcal{O}(n)$ and \eqref{eq:cartProdZonotope} has complexity $\mathcal{O}(1)$.  
	\label{prop:cartProduct}
\end{proposition}

\begin{proof}
	The result is obtained by inserting the definition of zonotopes \eqref{eq:zonotope} and \SPZsl \eqref{eq:polyZonotope} into the definition of the Cartesian product (see Sec. \ref{subsec:notation}).

Complexity: The construction of the resulting \SPZsl only involves concatenations and therefore has complexity $\mathcal{O}(1)$. For two \SPZs, $p_1 + p_2$ unique identifiers have to be generated, resulting in the complexity $\mathcal{O}(p_1 + p_2)$, which is $\mathcal{O}(n)$ using \eqref{eq:complexity}.
\end{proof}

% Quadratic Map

\subsubsection{Quadratic Map}

For reachability analysis based on the conservative polynomialization approach \cite{Althoff2013a}, a polynomial abstraction of the nonlinear dynamic function is calculated, requiring quadratic and higher-order maps. Here, we derive the equations for the quadratic map.

\begin{definition}
	(Quadratic Map) \cite[Theorem 1]{Althoff2013a} Given a set $\mathcal{S} \subset \mathbb{R}^n$ and a discrete set of matrices $Q_{i} \in \mathbb{R}^{n \times n}, i = 1 \dots m$, the quadratic map of $\mathcal{S}$ is defined as 
	\begin{equation*}
		\operator{sq}(Q,\mathcal{S}) = \left\{ x \left|~ x_{(i)} = s^T Q_{i} s,~ s \in \mathcal{S},~i = 1 \dots m \right. \right\}.
	\end{equation*}
	\label{def:quadMap}
\end{definition}

For \SPZs, we first consider the special case without independent generators, and later present the general case.
 
\begin{proposition}
Given the \SPZl $\widehat{PZ} \linebreak[1] = \langle \widehat{G},[~], \widehat{E}, \widehat{\ID} \rangle_{PZ}$ and a discrete set of matrices $Q_{i} \in \mathbb{R}^{n \times n}, i = 1 \dots m$, the result of the quadratic map is
	\begin{equation}
		\begin{split}
		& \operator{sq}(Q,\mathcal{\widehat{PZ}}) = \langle \overline{G}, [~], \overline{E}, \widehat{\ID} \rangle_{PZ} \\
		& ~~ \\
		& \mathrm{text} ~~ \overline{E} = \begin{bmatrix} \overline{E}_1 & \dots & \overline{E}_{\widehat{h}} \end{bmatrix},~ \overline{G} = \begin{bmatrix} \overline{G}_1 & \dots & \overline{G}_{\widehat{h}} \end{bmatrix},
	\end{split}
	\label{eq:quadMap}
	\end{equation}
where 
	\begin{equation*}
		\overline{E}_j = \widehat{E} + \widehat{E}_{(\cdot,j)} \cdot \mathbf{1}^{(1,\widehat{h})} , ~~ \overline{G}_j = \begin{bmatrix} \widehat{G}_{(\cdot,j)}^T Q_{1} \widehat{G} \\ \vdots \\ \widehat{G}_{(\cdot,j)}^T Q_{m} \widehat{G} \end{bmatrix}, ~ j = 1\dots \widehat{h}.
	\end{equation*}
	The \operator{compact} operation is applied to remove monomials with an identical variable part. The overall complexity is $\mathcal{O}(n^3 \log(n))+ \mathcal{O}(n^3 m)$.
	\label{prop:quadMap1}
\end{proposition}

\begin{proof}
	The equations are obtained by substitution of $s$ in Def.~\ref{def:quadMap} with the definition of an \SPZl from \eqref{eq:polyZonotope}, which yields
	\begin{equation*}
		\begin{split}
			& \operator{sq}(Q,\widehat{\mathcal{PZ}}) = \bigg\{ x ~ \bigg| ~ i = 1 \dots m,~ \alpha_k \in [-1,1],  \\
			& x_{(i)} = \bigg( \sum _{j=1}^{\widehat{h}} \prod _{k=1}^{\widehat{p}} \alpha _k ^{\widehat{E}_{(k,j)}}  \widehat{G}_{(\cdot,j)} \bigg)^T Q_i \bigg( \sum _{l=1}^{\widehat{h}} \prod _{k=1}^{\widehat{p}} \alpha _k ^{\widehat{E}_{(k,l)}}  \widehat{G}_{(\cdot,l)} \bigg) \bigg \} \\
			& ~ \\
			& = \bigg \{ x ~ \bigg | ~ x_{(i)} = \sum_{j = 1}^{\widehat{h}} \sum_{l = 1}^{\widehat{h}} \bigg( \prod _{k=1}^{\widehat{p}} \alpha _k ^{\widehat{E}_{(k,j)} + \widehat{E}_{(k,l)}} \bigg) \widehat{G}_{(\cdot,j)}^T Q_i \widehat{G}_{(\cdot,l)}, \\
			& ~~~~~~~~~~~ i = 1 \dots m,~ \alpha_k \in [-1,1] \bigg \} = \langle \overline{G}, [~], \overline{E}, \widehat{\ID} \rangle_{PZ}.
		\end{split}
	\end{equation*}	
Note that only the generator matrix, but not the exponent matrix, is different for each dimension $x_{(i)}$.
	
Complexity: The construction of the matrices $\overline{E}_j$ has complexity $\mathcal{O}(\widehat{h}^2 \widehat{p})$, and the construction of the matrices $\overline{G}_j$ has complexity $\mathcal{O}(n^2 \widehat{h} m) + \mathcal{O}(n \widehat{h}^2 m)$ if the results for $Q_i \widehat{G}$ are stored and reused. Since the matrices $\overline{E}$ and $\overline{G}$ both consist of $\overline{h} = \widehat{h}^2$ columns, and because the complexity of the \operator{compact} operation is $\mathcal{O}(\widehat{p}\overline{h} \log(\overline{h}))$ (see Prop.~\ref{prop:compact}), the complexity of the subsequent application of \operator{compact} is $\mathcal{O}(\widehat{p} \widehat{h}^2 \log(\widehat{h}^2))$. The resulting overall complexity is therefore $\mathcal{O}(n^2 \widehat{h} m) + \mathcal{O}(n \widehat{h}^2 m) + \mathcal{O}(\widehat{p} \widehat{h}^2 \log(\widehat{h}^2))$, which is $\mathcal{O}(n^3 \log(n))+ \mathcal{O}(n^3 m)$ using \eqref{eq:complexity}.
\end{proof}

We now extend Prop.~\ref{prop:quadMap1} to the general case including independent generators, for which we compute an over-approxi\-mation for computational reasons.

\begin{proposition}
	(Quadratic Map) Given an \SPZl $\mathcal{PZ} = \langle G,G_I,E,\ID \rangle_{PZ} \subset \mathbb{R}^n$ and a discrete set of matrices $Q_{i} \in \mathbb{R}^{n \times n}, i = 1 \dots m$,
	\begin{equation}
		\begin{split}
			& \operator{sq}(Q,\mathcal{PZ}) \subseteq \big\langle  \big[ c_z ~ \overline{G}_{(\cdot,\mathcal{H})} \big], G_z, \\
			& ~~~~~~~~~~~~~~~~~~~ \big[ \mathbf{0}^{(p,1)} ~ \overline{E}_{\left(\{ 1, \dots ,p \} , \mathcal{H} \right)} \big], \ID \big\rangle_{PZ} \\
			& ~ \\
			& \text{with} ~ \mathcal{K} = \left\{ i ~\left|~ \exists j > p ~ \overline{E}_{(j,i)} \neq 0 \right. \right\}, \\
			& ~~~~~~ \mathcal{H} = \{ 1, \dots , h+q \} \setminus \mathcal{K}, \\
			& ~~~~~~ \langle c_z, G_z \rangle_{Z} = \operator{zono} \left( \big\langle \overline{G}_{\left( \cdot, \mathcal{K} \right)}, [~], \overline{E}_{\left( \cdot, \mathcal{K} \right)}, \widehat{\ID} \big\rangle_{PZ} \right),
		\end{split}
		\label{eq:quadMapPolyZono}
	\end{equation} 
	where the zonotope enclosure is calculated by applying Prop.~\ref{prop:zonoEnclose}. The matrices $\overline{G}$ and $\overline{E}$ are computed according to \eqref{eq:quadMap} with the extended generator and exponent matrices $\widehat{G}$ and $\widehat{E}$, as well as the extended identifier vector $\widehat{\ID}$ defined as 
	\begin{equation}
		\begin{split}
		& \widehat{G} = \begin{bmatrix} G & G_I \end{bmatrix}, ~ \widehat{E} = \begin{bmatrix} E & \mathbf{0}^{(p,q)} \\ \mathbf{0}^{(q,h)} & I_q \end{bmatrix}, \\
		& \widehat{\ID} = \begin{bmatrix} \ID &  \operator{uniqueID}(q) \end{bmatrix},
		\end{split}
		\label{eq:quadMapExtMat}
	\end{equation}	
so that $\widehat{p} = p + q$. The \operator{compact} operation is applied to remove monomials with an identical variable part. The complexity of the calculations is $\mathcal{O}(n^3 \log(n)) + \mathcal{O}(n^3 m)$. 
	\label{prop:quadMapSpecial}
\end{proposition} 

\begin{proof}
	With the extended generator and exponent matrices $\widehat{G}$ and $\widehat{E}$ and the extended identifier vector $\widehat{\ID}$ from \eqref{eq:quadMapExtMat}, $\mathcal{PZ}$ can be represented equivalently as an \SPZl without independent generators
	\begin{equation*}
		\begin{split}
    \mathcal{PZ} = \bigg\{ & \sum _{i=1}^h \bigg( \prod _{k=1}^p \alpha _k ^{E_{(k,i)}} \bigg) G_{(\cdot,i)} + \sum _{j=1}^{q} \beta _j G_{I(\cdot,j)} ~ \bigg| \\ & \alpha_k, \beta_j \in [-1,1] \bigg\} \\
    & ~ \\
    = \bigg\{ & \sum _{i=1}^{h+q} \bigg( \prod _{k=1}^{p+q} \alpha _k ^{\widehat{E}_{(k,i)}} \bigg) \widehat{G}_{(\cdot,i)} ~ \bigg| ~ \alpha_k \in [-1,1] \bigg\} \\
    & ~ \\
    = \langle & \widehat{G},[~],\widehat{E},\widehat{id} \rangle_{PZ},
    \end{split}
	\end{equation*}
	which enables the computation of the quadratic map according to Prop.~\ref{prop:quadMap1}. For computational reasons, the resulting matrices $\overline{G}$ and $\overline{E}$ from Prop.~\ref{prop:quadMap1} are divided into a dependent part that contains the dependent factors $\alpha_1,\dots,\alpha_p$ only ($i \in \mathcal{H}$), and a remainder part that contains all remaining monomials ($i \in \mathcal{K}$):
	\begin{equation*}
	\begin{split}
		& \operator{sq}(Q,\langle \widehat{G},[~],\widehat{E},\widehat{id} \rangle_{PZ}) = \langle \overline{G}, [~], \overline{E}, \widehat{\ID} \rangle_{PZ} \\
		& ~ \\
		& = \bigg\{ \underbrace{ \sum _{i \in \mathcal{H}} \bigg( \prod _{k=1}^{p} \alpha _k ^{\overline{E}_{(k,i)}} \bigg) \overline{G}_{(\cdot,i)} }_{\text{dependent} ~ \text{part}} + \underbrace{ \sum _{i \in \mathcal{K}} \bigg( \prod _{k=1}^{p+q} \alpha _k ^{\overline{E}_{(k,i)}} \bigg) \overline{G}_{(\cdot,i)}}_{\text{remainder} ~ \text{part}} ~ \bigg| \\
		& ~~~~~~~ \alpha_k \in [-1,1] \bigg\}.
	\end{split}
	\end{equation*} 
	Since the part containing the remaining monomials is enclosed by a zonotope, it holds that the resulting \SPZl encloses the quadratic map.

Complexity: With the extended matrices $\widehat{G}$ and $\widehat{E}$ from \eqref{eq:quadMapExtMat}, the calculation of the matrices $\overline{E}_j$ and $\overline{G}_j$ has complexity $\mathcal{O}((h+q)^2 (p+q))$ and $\mathcal{O}(n^2 (h+q) m) + \mathcal{O}(n (h+q)^2 m)$, respectively, which is $\mathcal{O}(n^3 m)$ using \eqref{eq:complexity}. As for the case without independent generators, the complexity from the subsequent application of the \operator{compact} operation is $\mathcal{O}(ph^2 \log (h^2)) = \mathcal{O}(n^3 \log(n))$ using \eqref{eq:complexity}. The resulting overall complexity is therefore $\mathcal{O}(n^3 \log(n)) + \mathcal{O}(n^3 m)$, since all other operations have lower complexity.
\end{proof}

For the extension to cubic and higher-order maps of sets, which is omitted at this point due to space limitations, the results from Prop.~\ref{prop:quadMapSpecial} can be reused. We demonstrate the tightness of the quadratic map enclosure with an example:

\begin{example}
	We consider the \SPZ
	\begin{equation*}
		\mathcal{PZ} = \left\langle \begin{bmatrix} 1 & -1 & 1 \\ -1 & 2 & 1 \end{bmatrix}, \begin{bmatrix} 0.1 \\ 0 \end{bmatrix}, \begin{bmatrix} 1 & 0 & 2 \\ 0 & 1 & 1 \end{bmatrix}, [1 ~ 2 ] \right\rangle_{PZ}
	\end{equation*}
	and the matrices 
	\begin{equation*}
		Q_1 = \begin{bmatrix} 0.5 & 0.5 \\ 1 & -0.5\end{bmatrix}, ~~ Q_2 = \begin{bmatrix} -1 & 0 \\ 1 & 0 \end{bmatrix}.
	\end{equation*}
	A comparison between the exact quadratic map and the over-approximation computed by Prop.~\ref{prop:quadMapSpecial} is shown in Fig.~\ref{fig:exampleQuadMap}.
	\label{ex:quadMap}
\end{example} 

\begin{figure}[h]
\begin{center}
	\includegraphics[width = 0.45 \textwidth]{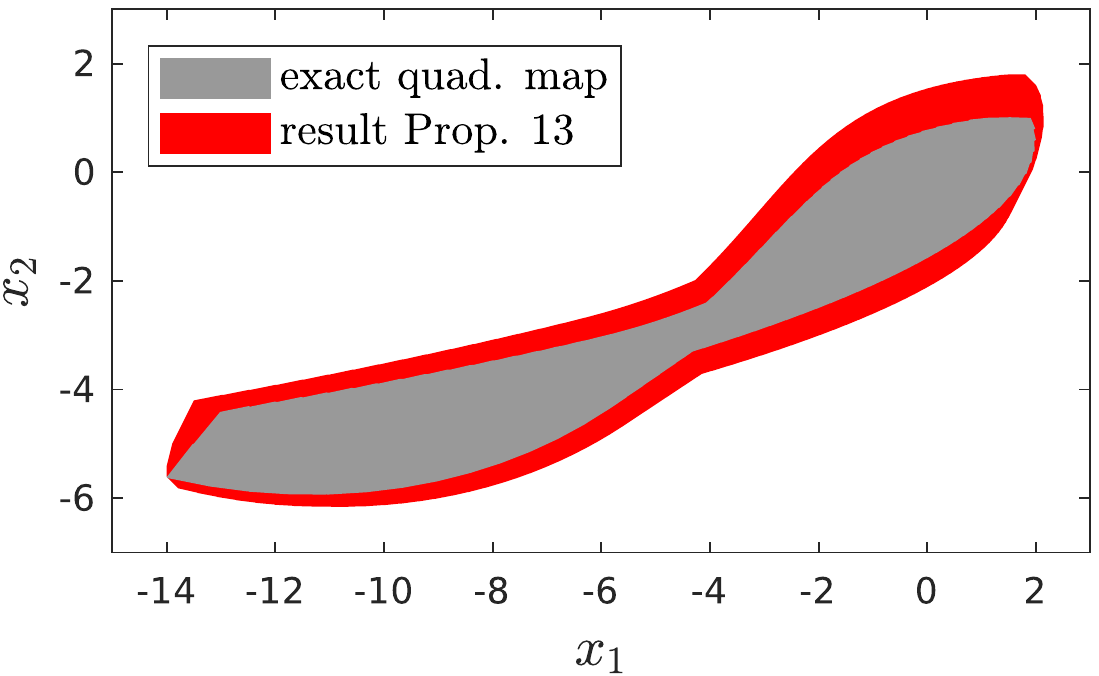}
	\caption{Exact quadratic map and over-approximation computed with Prop.~\ref{prop:quadMapSpecial} for the \SPZl in Example~\ref{ex:quadMap}.}
	\label{fig:exampleQuadMap}
	\end{center}
\end{figure}

% Convex Hull

\subsubsection{Convex Hull}

Another important set operation is the convex hull of two sets:

\begin{definition}
(Convex Hull) Given two sets $\mathcal{S}_1 \subset \mathbb{R}^n$ and $\mathcal{S}_2 \subset \mathbb{R}^n$, the convex hull of them is defined as 
\begin{equation*}
\begin{split}
\operator{conv}(\mathcal{S}_1,\mathcal{S}_2) = \big \{ & 0.5 \left( 1 + \lambda \right) s_1 + 0.5 \left( 1 - \lambda \right) s_2  ~\big| \\ 
& ~ s_1 \in \mathcal{S}_1,~s_2 \in \mathcal{S}_2,~\lambda \in [-1,1] \big \}.
\end{split}
\end{equation*}
\label{def:convexHull}
\end{definition}

We first consider the special case without independent generators and later present the general case.
\begin{proposition}
	Given $\mathcal{PZ}_1 = \langle G_1, [~], E_1,\linebreak[1] \ID_1 \rangle_{PZ}$ and $\mathcal{PZ}_2 = \langle G_2, [~],$ $E_2, \ID_2 \rangle_{PZ}$, the convex hull is computed as 
	\begin{equation}
		\begin{split}
			& \operator{conv}(\mathcal{PZ}_1, \mathcal{PZ}_2) = \langle \overline{G}, [~], \overline{E}, \overline{\ID} \rangle_{PZ}\\
			& ~~ \\
			& \text{with} ~~ \overline{G} = 0.5 \begin{bmatrix} G_1 & G_1 & G_2 & -G_2 \end{bmatrix}, \\
			& ~~~~~~~ \overline{E} = \begin{bmatrix} E_1 & E_1 & \mathbf{0}^{(p_1,h_2)} & \mathbf{0}^{(p_1,h_2)} \\ \mathbf{0}^{(p_2,h_1)} & \mathbf{0}^{(p_2,h_1)} & E_2 & E_2 \\ \mathbf{0}^{(1,h_1)} & \mathbf{1}^{(1,h_1)} & \mathbf{0}^{(1,h_2)} & \mathbf{1}^{(1,h_2)} \end{bmatrix}, \\
			& ~~~~~~~ \overline{\ID} = \operator{uniqueID}(p_1 + p_2 + 1),
		\end{split}
		\label{eq:convHull1}
	\end{equation}
which has complexity $\mathcal{O}(n^2)$.
\label{prop:convHull1}
\end{proposition}
\begin{proof}
	We first generate new unique identifiers for all dependent factors to remove possible dependencies between the two \SPZs. For \SPZs, the definition of the convex hull from Def.~\ref{def:convexHull} can be equivalently written as
	\begin{equation}
	\begin{split}
		\operator{conv}(&\mathcal{PZ}_1,\mathcal{PZ}_2) = \big \{ 0.5 \big ( \mathcal{PZ}_1 \boxplus \lambda ~ \mathcal{PZ}_1 \big ) \\
		& \oplus 0.5 \big ( \mathcal{PZ}_2 \boxplus (- \lambda )~ \mathcal{PZ}_2 \big ) ~ \big|~ \lambda \in [-1,1] \big\}.
	\end{split}
	\label{eq:convHullDef}
	\end{equation}	
	Evaluation of the exact additions and the Minkowski sum in \eqref{eq:convHullDef} for the \SPZsl according to Prop.~\ref{prop:exactAddition} and Prop.~\ref{prop:addition} results in the equations in \eqref{eq:convHull1}, where we substituted the parameter $\lambda$ with an additional dependent factor $\alpha_{p_1 + p_2 + 1} = \lambda$. Since $\lambda \in [-1,1]$ and $\alpha_{p_1 + p_2 + 1} \in [-1,1]$, this substitution does not change the set. 
	
Complexity: The construction of the matrix $\overline{G}$ has complexity $\mathcal{O}(2n(h_1 + h_2))$. Generation of $p_1 + p_2 + 1$ unique identifiers for $\overline{\ID}$ has complexity $\mathcal{O}(p_1 + p_2 + 1)$. The overall complexity is therefore $\mathcal{O}(2n(h_1 + h_2)) + \mathcal{O}(p_1 + p_2 + 1)$, which is $\mathcal{O}(n^2)$ using \eqref{eq:complexity}. 	  
\end{proof}

We now extend Prop.~\ref{prop:convHull1} to the general case including independent generators, for which we compute an over-approxi\-mation for computational reasons.
\begin{proposition}
	(Convex Hull) Given the two \SPZsl $\mathcal{PZ}_1 = \langle G_1, G_{I,1}, E_1, \ID_1 \rangle_{PZ}$ and $\mathcal{PZ}_2 = \langle G_2, G_{I,2},$ $E_2, \ID_2 \rangle_{PZ}$,
	\begin{equation*}
		\begin{split}
			& \operator{conv}(\mathcal{PZ}_1, \mathcal{PZ}_2) \subseteq \langle \overline{G}, \overline{G}_{I}, \overline{E}, \overline{\ID} \rangle_{PZ} \\
			& ~~ \\
			& \text{with} ~ \langle \overline{G}, [~], \overline{E}, \overline{\ID} \rangle_{PZ} = \operator{conv}(\overline{\mathcal{PZ}}_1, \overline{\mathcal{PZ}}_2 ), \\
			& ~~~~~~~  \langle \mathbf{0}^{(n,1)}, \overline{G}_I \rangle_Z \supseteq \operator{conv}(\langle \mathbf{0}^{(n,1)}, G_{I,1} \rangle_Z , \langle \mathbf{0}^{(n,1)}, G_{I,2} \rangle_Z), \\
			& ~~~~~~~ \overline{\mathcal{PZ}}_1 = \langle G_1,[~],E_1,\ID_1 \rangle_{PZ}, \\
			& ~~~~~~~ \overline{\mathcal{PZ}}_2 = \langle G_2,[~],E_2,\ID_2 \rangle_{PZ},
		\end{split}
	\end{equation*}
	where $\operator{conv}(\overline{\mathcal{PZ}}_1, \overline{\mathcal{PZ}}_2)$ is calculated according to Prop.~\ref{prop:convHull1}, and an over-approximation of the convex hull of two zonotopes is computed according to \cite[Eq. (2.2)]{Althoff2010a}:
	\begin{equation}
		\begin{split}
			& \operator{conv}(\langle \mathbf{0}^{(n,1)}, G_{I,1} \rangle_Z , \langle \mathbf{0}^{(n,1)}, G_{I,2} \rangle_Z) \subseteq \langle \mathbf{0}^{(n,1)}, \overline{G}_I \rangle_Z \\
			& ~~ \\
			& \text{with} ~~ \overline{G}_I = \begin{cases} [ \widehat{G}_1 ~~ G_{I,1(\cdot,\{q_2+1, \dots ,q_1\})} ], ~~ q_1 \geq q_2 \\ [ \widehat{G}_2 ~~ G_{I,2(\cdot,\{q_1+1, \dots, q_2\})} ], ~~ q_1 < q_2 \end{cases}, \\
			& ~~~~~~~ \widehat{G}_1 = \frac{1}{2} \begin{bmatrix} G_{I,1(\cdot,\mathcal{K}_2)} + G_{I,2} & G_{I,1(\cdot,\mathcal{K}_2)} - G_{I,2} \end{bmatrix}, \\
			& ~~~~~~~ \widehat{G}_2 = \frac{1}{2} \begin{bmatrix} G_{I,1} + G_{I,2(\cdot,\mathcal{K}_1)} & G_{I,1} - G_{I,2(\cdot,\mathcal{K}_1)} \end{bmatrix},
		\end{split}
		\label{eq:convHullZono}
	\end{equation}
	where $\mathcal{K}_1 = \{ 1,\dots,q_1 \}$ and $\mathcal{K}_2 = \{ 1,\dots,q_2 \}$. The overall complexity is $\mathcal{O}(n^2)$.
	\label{prop:convHullSpecial}
\end{proposition}

\begin{proof}
	Each \SPZl $\mathcal{PZ} = \langle G,G_I,E,\ID \rangle_{PZ}$ can be represented equivalently as the Minkowski sum of an \SPZl and a zonotope $\mathcal{PZ} = \langle G,[~],E,\ID \rangle_{PZ} \oplus \langle \mathbf{0}^{(n,1)}, G_I \rangle_Z$ (see \eqref{eq:additionZonotope}). Given sets $S_1$, $S_2$, $S_3$, $S_4 \subset \mathbb{R}^n$, it holds that $\operator{conv}(S_1 \oplus S_2, S_3 \oplus S_4) \subseteq \operator{conv}(S_1,S_3) \oplus \operator{conv}(S_2,S_4)$, which can be derived from the definition of the convex hull from Def.~\ref{def:convexHull}:
\begin{equation*}
	\begin{split}
		& \operator{conv}( S_1 \oplus S_2, S_3 \oplus S_4) \\ 
		& ~ \\
		& = \big \{ 0.5 (1 + \lambda) (s_1 + s_2) + 0.5 (1-\lambda) (s_3 + s_4) ~ \big | \\
		&  ~~~~~~ s_1 \in \mathcal{S}_1, s_2 \in \mathcal{S}_2, s_3 \in \mathcal{S}_3, s_4 \in \mathcal{S}_4, \lambda \in [-1,1] \big \}
	\end{split}
\end{equation*}
\begin{equation*}
\begin{split}
		& = \big \{ \underbrace{0.5 (1 + \lambda) s_1 + 0.5 (1-\lambda) s_3}_{\operator{conv}(\mathcal{S}_1,\mathcal{S}_3)} \\ 
		& ~~~~~ + \underbrace{0.5 (1 + \lambda) s_2 + 0.5 (1-\lambda) s_4}_{\operator{conv}(\mathcal{S}_2,\mathcal{S}_4)}  ~ \big | \\
		&  ~~~~~~ s_1 \in \mathcal{S}_1, s_2 \in \mathcal{S}_2, s_3 \in \mathcal{S}_3, s_4 \in \mathcal{S}_4, \lambda \in [-1,1] \big \} \\
		& ~ \\
		& \subseteq \big \{ 0.5 (1 + \lambda) s_1 + 0.5 (1-\lambda) s_3 ~\big | \\
		& ~~~~~ s_1 \in \mathcal{S}_1, s_3 \in \mathcal{S}_3, \lambda \in [-1,1] \big \} \\ 
		& ~\oplus \big \{ 0.5 (1 + \lambda) s_2 + 0.5 (1-\lambda) s_4 ~ \big | \\
		& ~~~~~~ s_2 \in \mathcal{S}_2, s_4 \in \mathcal{S}_4, \lambda \in [-1,1] \big \}.
	\end{split}
\end{equation*}	
An enclosure of the convex hull of two \SPZsl can therefore be obtained by computing the convex hull of two \SPZsl without independent generators according to Prop.~\ref{prop:convHull1}, and the computation of the convex hull of two zonotopes according to \cite[Eq. (2.2)]{Althoff2010a}.
	
Complexity: As shown in Prop.~\ref{prop:convHull1}, the complexity for computing the convex hull of two \SPZsl without independent generators is $\mathcal{O}(n^2)$. The complexity for the computation of the convex hull of two zonotopes as defined in \eqref{eq:convHullZono} is $\mathcal{O}(2n \min(q_1,q_2))$, resulting from the matrix additions and subtractions that are required for the construction of the matrix $\widehat{G}_1$ or $\widehat{G}_2$. The overall complexity is therefore $\mathcal{O}(n^2) + \mathcal{O}(2n \min(q_1,q_2))$, which is $O(n^2)$ using \eqref{eq:complexity}.
\end{proof}

We demonstrate the tightness of the convex hull enclosure with an example:

\begin{example}
	We consider the \SPZ s
	\begin{equation*}
		\mathcal{PZ}_1 = \left\langle \begin{bmatrix} -2 & 2 & 0 & 1 \\ -2 & 0 & 2 & 1 \end{bmatrix}, [~], \begin{bmatrix} 0 & 1 & 0 & 3 \\ 0 & 0 & 1 & 1 \end{bmatrix}, [1 ~ 2 ] \right\rangle_{PZ}
	\end{equation*}
	and 
	\begin{equation*}
		\mathcal{PZ}_2 = \left\langle \begin{bmatrix} 3 & 1 & -2 & 1 \\ 3 & 2 & 3 & 1 \end{bmatrix}, \begin{bmatrix} 0.5 \\ 0 \end{bmatrix}, \begin{bmatrix} 0 & 1 & 0 & 2 \\ 0 & 0 & 1 & 1 \end{bmatrix}, [1 ~ 2 ] \right\rangle_{PZ}.
	\end{equation*}
	A comparision between the exact convex hull and the over-approximation computed by Prop.~\ref{prop:convHullSpecial} is shown in Fig.~\ref{fig:exampleConvHull}.
	\label{ex:convHull}
\end{example}

\begin{figure}[h]
\begin{center}
	\includegraphics[width = 0.45 \textwidth]{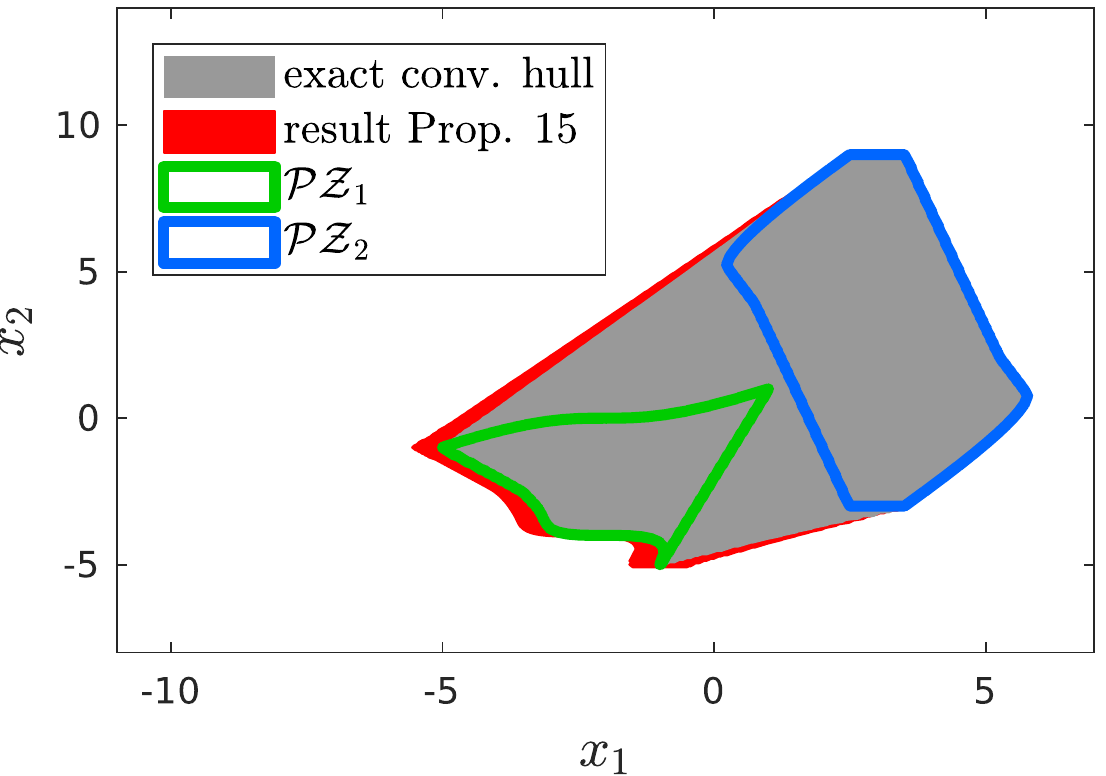}
	\caption{Exact convex hull and over-approximation computed with Prop.~\ref{prop:convHullSpecial} for the \SPZ s in Example~\ref{ex:convHull}.}
	\label{fig:exampleConvHull}
	\end{center}
\end{figure}

% Auxiliary operations

\subsection{Auxiliary Operations}
\label{subsec:auxiliaryOperations}

This subsection derives useful auxiliary operations on \SPZs.

% Order reduction

\subsubsection{Order Reduction}

Many set operations, such as Minkowski addition or quad\-ratic maps, increase the number of generators and consequently also the order $\rho$ of the \SPZ. Thus, for computational reasons, it is necessary to repeatedly reduce the zonotope order during reachability analysis. We propose a reduction operation for \SPZsl that is based on the order reduction of zonotopes (see e.g., \cite{Kopetzki2017}).

\begin{proposition}
	(Reduce) Given an \SPZl $\mathcal{PZ} = \langle G,G_I,\linebreak[1] E,\ID \rangle_{PZ}$ and a desired zonotope order $\rho_d \geq 1 + 1/n$, the operation \operator{reduce} returns an \SPZl with an order smaller than or equal to $\rho_d$ that encloses $\mathcal{PZ}$:
	\begin{equation*}
 		\begin{split}
			& \operator{reduce}(\mathcal{PZ},\rho_d) = \bigg\langle \begin{bmatrix} c_z & G_{(\cdot,\widehat{\mathcal{K}})} \end{bmatrix}, \begin{bmatrix} G_{I(\cdot,\widehat{\mathcal{H}})} & G_z \end{bmatrix}, \\ 
			& ~~~~~~~~~~~~~~~~~~~~~~~~~~~  \begin{bmatrix} \mathbf{0}^{(p,1)} & E_{(\cdot,\widehat{\mathcal{K}})} \end{bmatrix}, ~ \ID \bigg\rangle_{PZ} \\
			& ~~ \\
			& \text{with} ~~ \langle c_z, G_z \rangle_Z = \operator{reduce}(\mathcal{Z},1), \\
			& ~~~~~~~~\mathcal{Z} = \operator{zono} \big( \big\langle G_{(\cdot,\mathcal{K})},~ G_{I(\cdot,\mathcal{H})},~ E_{(\cdot,\mathcal{K})},~ \ID \big\rangle_{PZ} \big),
		\end{split}
	\end{equation*}
	where the zonotope enclosure is calculated by applying Prop.~\ref{prop:zonoEnclose}. For reduction, the
	\begin{equation}
	 	a = \max \left(0,\min \left( h+q,\lceil h + q - n (\rho_d-1) + 1\rceil \right) \right)
	 	\label{eq:nrOfSelGen}
	\end{equation} 
	smallest generators are selected:
	\begin{equation*}
	\begin{split}
		& \mathcal{K} = \begin{cases} \emptyset, & \mathrm{if} ~ a = 0 \\ \left\{ i ~\left| ~ {|| G_{(\cdot,i)} ||}_2 \leq {|| \widehat{G}_{(\cdot,d_{(a)})} ||}_2 \right. \right\}, & \mathrm{otherwise} \end{cases}, \\
		& \mathcal{H} = \begin{cases} \emptyset, & \mathrm{if} ~ a = 0 \\\left\{ i ~\left| ~ {|| G_{I(\cdot,i)} ||}_2 \leq {|| \widehat{G}_{(\cdot,d_{(a)})} ||}_2 \right. \right\}, & \mathrm{otherwise} \end{cases},\\
		& \widehat{\mathcal{K}} = \{1, \dots ,h \} \setminus \mathcal{K},~~ \widehat{\mathcal{H}} = \{1, \dots ,q \} \setminus \mathcal{H}, \\
		& ~~ \\
		& \text{with} ~~ {|| \widehat{G}_{(\cdot,d_{(1)})} ||}_2 \leq \dotsc \leq {||\widehat{G}_{(\cdot,d_{(h+q)})}||}_2~,
	\end{split}
	\end{equation*} 
	where $\widehat{G} = [G~G_I]$ and $d \in \mathbb{N}_{>0}^{h+q}$ is a vector of indices for which each next entry represents a longer generator than the previous entry. The complexity is $\mathcal{O}(n^2) + \mathcal{O}(\operator{reduce})$, where $\mathcal{O}(\operator{reduce})$ denotes the complexity of the zonotope reduction, which depends on the selected method. 
\end{proposition}

\begin{proof}
	The definition of $a$ (see \eqref{eq:nrOfSelGen}) ensures that $|\widehat{\mathcal{K}}| + |\widehat{\mathcal{H}}| + n + 1 \leq \rho_d n$ for values $\rho_d \geq 1 + 1/n$. Furthermore, $\operator{reduce}(\mathcal{PZ},\rho_d) \supseteq \mathcal{PZ}$ since the operators \operator{zono} and \operator{reduce} are both over-approximative, and therefore $\operator{reduce}(\operator{zono}(\cdot))$ is over-$\linebreak[1]$approximative, too. 
	
Complexity: Sorting the generators has a complexity of $\mathcal{O}(n(h+q)) + \mathcal{O}((h+q) \log(h+q))$, which is $\mathcal{O}(n^2)$ using \eqref{eq:complexity}. In the worst-case where all dependent generators get reduced, the enclosure with a zonotope has complexity $\mathcal{O}(p h) + \mathcal{O}(n h)$ (see. Prop.~\ref{prop:zonoEnclose}), which is $\mathcal{O}(n^2)$ using \eqref{eq:complexity}. The overall complexity is therefore $\mathcal{O}(n^2) \linebreak[1] + \mathcal{O}(\operator{reduce})$.
\end{proof}

After reduction, we remove possibly generated all-zero rows in the exponent matrix.

% Restructure

\subsubsection{Restructure}

\label{subsubsec:restructure}

Due to the repeated order reduction and Minkowski addition during reachability analysis, the volume spanned by independent generators grows relative to the volume spanned by dependent generators. As explained later in Sec. \ref{sec:ReachabilityAnalysis}, this has a negative effect on the tightness of the reachable sets. We therefore define the operation \operator{restructure}, which introduces new dependent generators that over-approximate the independent ones:
\begin{proposition}
	(Restructure) Given an \SPZl $\mathcal{PZ} = \langle G, \linebreak[1] G_I,E,\ID \rangle_{PZ}$, \operator{restructure} returns an \SPZl that encloses $\mathcal{PZ}$ and removes all independent generators:
	\begin{equation}
		\begin{split}
			& \operator{restructure}(\mathcal{PZ}) = \left\langle \begin{bmatrix} c_z & G & G_z \end{bmatrix}, [~], \overline{E}, \overline{\ID} \right\rangle_{PZ} \\
			& ~~ \\
			& \text{with} ~~ \langle c_z,G_z\rangle_Z = \operator{reduce}(\langle \mathbf{0}^{(n,1)},G_I \rangle_Z, 1), \\
			& ~~~~~~~~~ \overline{E} = \begin{bmatrix} \mathbf{0}^{(p,1)} & E & \mathbf{0}^{(p,n)} \\ \mathbf{0}^{(n,1)} & \mathbf{0}^{(n,h)} & I_n \end{bmatrix}, \\
			& ~~~~~~~~~ \overline{\ID} = \begin{bmatrix} \ID & \operator{uniqueID}(n) \end{bmatrix}.
		\end{split}
		\label{eq:restructure}
	\end{equation}
The overall complexity is $\mathcal{O}(n) + \mathcal{O}(\operator{reduce})$, where $\mathcal{O}(\operator{reduce})$ is the complexity of the zonotope reduction.
 \label{prop:restructure}
\end{proposition}

\begin{proof}
	The result of the \operator{restructure} operation encloses the original set since \operator{reduce} is over-approx\-imative, and the redefinition of independent generators as new dependent generators just changes the set representation, but not the set in \eqref{eq:restructure} itself:
	\begin{equation*}
  	\begin{split}
    & \operator{restructure}(\mathcal{PZ}) = \bigg\{ c_z + \sum _{i=1}^h \bigg( \prod _{k=1}^p \alpha _k ^{E_{(k,i)}} \bigg) G_{(\cdot,i)} \\ 
    & ~~~~~~~~~~~~~~~~~~~~~~~~~~~~~ + \sum _{j=1}^{n} \beta _j G_{z(\cdot,j)} ~ \bigg| ~ \alpha_k, \beta_j \in [-1,1] \bigg\} \\
    & ~ \\
    & = \bigg\{   \sum _{i=1}^{h+n+1} \bigg( \prod _{k=1}^{p+n} \alpha _k ^{\overline{E}_{(k,i)}} \bigg) \overline{G}_{(\cdot,i)} ~ \bigg| ~ \alpha_k \in [-1,1] \bigg\}.
    \end{split}
  \end{equation*}
	
Complexity: The generation of $n$ new unique identifiers has complexity $\mathcal{O}(n)$. Since the construction of the matrices only involves concatenations, the overall complexity equals $\mathcal{O}(n) + \mathcal{O}(\operator{reduce})$, where $\mathcal{O}(\operator{reduce})$ is the complexity of the zonotope reduction. 
\end{proof}

We demonstrate the effectiveness of Prop.~\ref{prop:restructure} by numerical examples in Sec. \ref{sec:NumericalExamples}. To save computation time, we define an upper bound $p_d$ of factors for the \SPZl after restructuring so that $p+n \leq p_d$ holds, where independent factors are removed first to maintain as much dependence as possible. If required, dependent factors are removed by enclosing the corresponding generators with a zonotope (see Prop.~\ref{prop:zonoEnclose}).

% REACHABILITY ANALYSIS -------------------------------------

\section{Reachability Analysis}
\label{sec:ReachabilityAnalysis}

In this section, we demonstrate how \SPZsl can be used to improve reachability analysis for nonlinear systems. We consider nonlinear systems of the form
\begin{equation}
	\dot x (t) = f(x(t),u(t)), ~ x(t) \in \mathbb{R}^n,~ u(t) \in \mathbb{R}^m,
	\label{eq:nonlinSys}
\end{equation}
where $x$ is the vector of system states, $u$ is the input vector, and $f:~\mathbb{R}^n \times \mathbb{R}^m \to \mathbb{R}^n$ is a Lipschitz continuous function. The reachable set of the system is defined as follows:

\begin{definition}
	(Reachable Set) Let $\xi(t,x_0,u(\cdot))$ denote the solution to \eqref{eq:nonlinSys} for an initial state $x(0) = x_0$ and the input trajectory $u(\cdot)$. The reachable set for an initial set $\mathcal{X}_0 \subset \mathbb{R}^n$ and a set of possible input values $\mathcal{U} \subset \mathbb{R}^m$ is 
	\begin{equation*}
		\mathcal{R}_{\mathcal{X}_0}^e(t) := \big\{ \xi(t,x_0,u(\cdot)) ~\big |~ x_0 \in \mathcal{X}_0, \forall \tau \in [0,t]~ u(\tau) \in \mathcal{U} \big \}.
	\end{equation*} 
\end{definition}

The superscript $e$ on $\mathcal{R}_{\mathcal{X}_0}^e(t)$ denotes the exact reachable set, which cannot be computed for general nonlinear systems. Therefore, we compute a tight over-approximation $\mathcal{R}(t) \supseteq \mathcal{R}_{\mathcal{X}_0}^e(t)$. Furthermore, we calculate the reachable set for consecutive time intervals $\tau_s = [t_s,t_{s+1}]$ with $t_{s+1} = t_s + \Delta t$ so that the reachable set for a time horizon $t_f$ is given as $\mathcal{R}([0,t_f]) = \bigcup_{s = 0}^{t_f/\Delta t -1} \mathcal{R}(\tau_s)$, where $t_f$ is a multiple of  $\Delta t$.

\subsection{Reachability Algorithm}

Our algorithm in based on the conservative polynomialization approach in \cite{Althoff2013a}, which is an extension to the conservative linearization approach in \cite{Althoff2008c}. The principle of conservative linearization is quite simple: In each time interval $\tau_s$ the nonlinear function $f(\cdot)$ is linearized and the set of linearization errors is treated as an additional uncertain input to the system, which then allows to calculate the reachable set with a reachability algorithm for linear systems. The conservative polynomialization approach improves this method by abstracting the nonlinear function $f(\cdot)$ by a Taylor expansion of order $\kappa$:
\begin{equation*}
	\begin{split}
		\dot x_{(i)}(t) & = f_{(i)}(z(t)) \\
		& \in \sum_{j=0}^{\kappa} \frac{\left((z(t)-z^*)^T \nabla \right)^j f_{(i)}(z^*)}{j!} \oplus \mathcal{L}_{(i)}(t),
	\end{split}
\end{equation*} 
where $z = [x^T~u^T]^T$, the set $\mathcal{L}_{(i)}(t)$ is the Lagrange remainder, defined in \cite[Eq. (2)]{Althoff2013a}, and the vector $z^* \in \mathbb{R}^{n+m}$ is the expansion point for the Taylor series.

In order to fully exploit the advantages of \SPZs, Alg. \ref{alg:reach} is slightly modified from \cite{Althoff2013a}. We only specify the algorithm for the Taylor order $\kappa = 2$ for simplicity, since the extension to higher orders is straightforward. The operation \operator{taylor} returns the matrices
\begin{equation*}
	\begin{split}
	& A_{(i,j)} = \frac{\partial f_{(i)}(\cdot)}{\partial x_{(j)}} \bigg |_{z^*}\in \mathbb{R}^{n \times n},~B_{(i,j)} = \frac{\partial f_{(i)}(\cdot)}{\partial u_{(j)}} \bigg |_{z^*}\in \mathbb{R}^{n \times m},  \\
	& ~ \\
	& ~~~~~~~~ D = \nabla^2 f(z^*),~~E = \nabla^3 f(z^*),~~w = f(z^*)
	\end{split}
\end{equation*}
storing the coefficients of the Taylor series expansion for the nonlinear function $f(\cdot)$ at the expansion point $z^*$, and the operation \operator{enlarge} enlarges a set by a given scalar factor $\lambda \in \mathbb{R}_{> 1}$. The definitions of the operations $\operator{post}^{\Delta}$ \cite[Sec. 4.1]{Althoff2013a}, \operator{varInputs} \cite[Sec. 4.2]{Althoff2013a}, and \operator{lagrangeRemainder} \cite[Sec. 4.1]{Althoff2013a} are identical to the ones in \cite{Althoff2013a}. The \texttt{post} operator changed since we precompute some of the sets in our algorithm, and since we use the exact addition as defined in Prop.~\ref{prop:exactAddition} instead of the Minkowski addition to add $\mathcal{F}_1$ to $\mathcal{F}_2$: 

\begin{equation}
	\begin{split}
		& \text{\texttt{post}}\left(\mathcal{R}(t_s),A,\mathcal{V}(t_s),\mathcal{V}^{\Delta}(\tau_s),\mathcal{L}(\tau_s)\right) = \\
		& ~ \\
		& ~~~ \underbrace{e^{A \Delta t} \mathcal{R}(t_s)}_{\mathcal{F}_1} \boxplus \underbrace{\Gamma(\Delta t) \mathcal{V}(t_s)}_{\mathcal{F}_2} \oplus \mathcal{R}^{p,\Delta}\left(\mathcal{V}^{\Delta}(\tau_s) \oplus \mathcal{L}(\tau_s),\Delta t \right),
	\end{split}
	\label{eq:post}
\end{equation}
where $\Gamma(\Delta t)$ is defined as in \cite[Sec.~3.2]{Althoff2013a}, and $\mathcal{R}^{p,\Delta}(\cdot)$ is defined as in \cite[Eq.~(9)]{Althoff2013a}. We heuristically trigger the restructure process (see Sec. \ref{subsubsec:restructure}) when 
\begin{equation*}
	\begin{split}
	& \operator{volRatio}(\mathcal{PZ}) \\
	& ~~ = \frac{\operator{volume}(\operator{interval}(\langle \mathbf{0}^{(n,1)}, G_I \rangle_Z))}{\operator{volume}(\operator{interval}(\operator{zono}(\langle G,[~],E,\ID \rangle_{PZ})))} > \mu_d,
	\end{split}
\end{equation*}
where $\mathcal{PZ} = \langle G, G_I,E,\ID \rangle_{PZ}$, \operator{interval} computes an interval enclosure, and \operator{volume} calculates the volume of a multi-dimensional interval. 

\begin{algorithm}[h!tb]
	\caption{Compute a tight enclosure of the reachable set} \label{alg:reach}
	\textbf{Require:} Initial set $\mathcal{R}(0)$, input set $\mathcal{U}$, time horizon $t_f$,
	
	$~~~~~~~~~~~$ time step size $\Delta t$, enlargement factor $\lambda$, maximum 
	
	$~~~~~~~~~~~$ zonotope order $\rho_d$, maximum volume ratio $\mu_d$.
	
	\textbf{Ensure:} Reachable set $\mathcal{R}([0,t_f])$.
	\begin{algorithmic}[1]
		\State $t_0 = 0,~ s = 0,~ \mathcal{R}^{union} = \emptyset,~\mathcal{U}_s = \{ \mathbf{0}^{(m,1))} \} $
		\State $\Psi(\tau_0) = \{ \mathbf{0}^{(n,1)} \}$
		\While{$t_s < t_f$} \label{line:startWhile}
			\State \texttt{taylor} $\rightarrow w, A, B ,D ,E, ~ z^* = [x^{*T}~u^{*T}]^T$ \label{line:taylor}
			\State $\mathcal{R}^d(t_s) = \mathcal{R}(t_s) \oplus (-x^*),~\mathcal{U}^{\Delta} = \mathcal{U} \oplus (-u^*)$ \label{line:startLinError}
			\State $\mathcal{V}(t_s) = \{w - A x^* \} \oplus \frac{1}{2}sq(D,\mathcal{R}^d(t_s) \times \mathcal{U}_s )$ \label{line:quadMap}
			\State $\mathcal{Z}_z(t_s) = \operator{zono}(\mathcal{R}^d(t_s)) \times \mathcal{U}_s$
			\Repeat
				\State $\overline{\Psi} (\tau _s) = $~\texttt{enlarge}$(\Psi(\tau_s),\lambda)$ \label{line:enlarge}
				\State $\mathcal{R}^{\Delta}(\tau_s) = $~\texttt{post}$^{\Delta}(\mathcal{R}(t_s),\overline{\Psi} (\tau _s),A)$
				\State $\mathcal{R}_z^{\Delta}(\tau_s) = \operator{zono}(\mathcal{R}^{\Delta}(\tau_s))$
				\State $\mathcal{V}^{\Delta}(\tau_s) = $~\texttt{varInputs}$(${\small$\mathcal{Z}_z(t_s),\mathcal{R}_z^{\Delta}(\tau_s),\mathcal{U}^{\Delta},B,D$}$)$
				\State $\mathcal{R}(\tau_s) = \mathcal{R}(t_s) \oplus \mathcal{R}_z^{\Delta}(\tau_s)$ \label{line:reachInt}
				\State $\mathcal{L}(\tau_s) = $~\texttt{lagrangeRemainder}$(\mathcal{R}(\tau_s),E,z^*)$
				\State $\Psi(\tau_s) = \mathcal{V}(t_s) \oplus \mathcal{V}^{\Delta}(\tau_s) \oplus \mathcal{L}(\tau_s)$ \label{line:linError}
			\Until{$\Psi(\tau_s) \subseteq \overline{\Psi} (\tau _s)$} \label{line:endLinError}
			\State $\mathcal{R}(t_{s+1}) = $~\texttt{post}$(\mathcal{R}(t_s),A,\mathcal{V}(t_s),\mathcal{V}^{\Delta}(\tau_s),\mathcal{L}(\tau_s))$ \label{line:post}
			\State $\mathcal{R}(t_{s+1}) = \operator{reduce}(\mathcal{R}(t_{s+1}),\rho_d)$ \label{line:reduce}
			\If{$\operator{volRatio}(\mathcal{R}(t_{s+1})) > \mu_d$}
				\State $\mathcal{R}(t_{s+1}) = \operator{restructure}(\mathcal{R}(t_{s+1}))$\label{line:restructure}
			\EndIf
			\State $\mathcal{R}^{union} = \mathcal{R}^{union} \cup \mathcal{R}(\tau_s)$
			\State $t_{s+1} = t_s + \Delta t, ~\Psi(\tau_{s+1}) = \Psi(\tau_s),~ s = s+ 1$ \label{line:linErrorUpdate}
		\EndWhile \label{line:endWhile}
		\State $\mathcal{R}([0,t_f]) = \mathcal{R}^{union}$
	\end{algorithmic}
\end{algorithm}

The while-loop in lines~\ref{line:startWhile}-\ref{line:endWhile} of Alg.~\ref{alg:reach} iterates over all time intervals $\tau_s$ until the time horizon $t_f$ is reached. For each time interval we first abstract the nonlinear equation $f(\cdot)$ by a Taylor expansion in Line~\ref{line:taylor}. In lines~\ref{line:startLinError}-\ref{line:endLinError} we then compute the set of linearization errors $\Psi(\tau_s)$ on the time interval reachable set $\mathcal{R}(\tau_s)$. The problem we are facing here is that we need $\mathcal{R}(\tau_s)$ to calculate $\Psi(\tau_s)$, but on the other hand we also need $\Psi(\tau_s)$ to calculate $\mathcal{R}(\tau_s)$. To resolve this mutual dependence we first compute $\mathcal{R}(\tau_s)$ using the set of linearization errors from the previous time step ($\Psi(\tau_{s+1}) = \Psi(\tau_{s})$, see Line~\ref{line:linErrorUpdate}) as an initial guess in Line~\ref{line:reachInt}. Next, we use $\mathcal{R}(\tau_s)$ to calculate $\Psi(\tau_s)$ in Line~\ref{line:linError}. In the next iteration of the repeat-until loop we then enlarge the set of linearization errors in Line~\ref{line:enlarge} and calculate $\mathcal{R}(\tau_s)$ and $\Psi(\tau_s)$ using the enlarged set $\overline{\Psi}(\tau_s)$. We repeat this process until the enlarged set of linearization errors $\overline{\Psi}(\tau_s)$ contains the set of actual linearization erros $\Psi(\tau_s)$ (see Line~\ref{line:endLinError}), which guarantees that $\mathcal{R}(\tau_s)$ and $\Psi(\tau_s)$ are both over-approximative. For the benchmarks in Sec.~\ref{sec:NumericalExamples} this procedure results in 1 to 3 iterations of the repeat-until loop. After we computed the set of linearization errors, we finally calculate the reachable set $\mathcal{R}(t_{s+1})$ for the next point in time in Line~\ref{line:post}, which we then use as the new initial set for the next time interval $\tau_{s+1}$. 

At the end of each time step, we apply the operation \operator{reduce} in Line~\ref{line:reduce} to reduce the zonotope order to the desired order $\rho_d$. Since the zonotope order is defined as $\rho = \frac{h+q}{n}$ (see Sec.~\ref{sec:SparsePolynomialZonotope}), this ensures that the SPZs that represent the reachable set contain at most $\rho_d  n$ generators, where $n$ is the system dimension. While order reduction limits the growth of the size of the exponent matrix, the growth of the integer entries in the exponent matrix is not contolled explicitly. However, since the generators that belong to large exponents are usually small they are automatically removed by the operation \operator{reduce}.

Since Alg.~\ref{alg:reach} contains several tuning parameters we shortly discuss their influence on the performance. The parameter with the largest effect is the time step size $\Delta t$; a smaller time step size improves the accuracy, but also prolongs the computation time. In addition, the maximum volume ratio $\mu_d$ can have a large impact. A smaller ratio potentially increases the accuracy, but a very small ratio can have a negative effect since every restructure operation results in an over-approximation. All other parameters have a smaller influence on the performance. Good default values are $\lambda = 0.1$ for the enlargement factor, $\rho_d = 50$ for the maximum zonotope order, and $p_d = 50$ for the maximum number of dependent factors. Furthermore, the method in \cite[Sec. 3.4]{Girard2005} (Girard's method) is applied for zonotope reduction, and we use principal-component-analysis-based order reduction in combination with Girard's method \cite[Sec.~III.A]{Kopetzki2017} for the reduction during the restructure operation. For the future we aim to develop strategies to tune the parameters automatically. 

We proceed with a discussion of the main advantages resulting from \SPZs.

% Advantages Sparse Polynomial Zonotopes

\subsection{Advantages of using Sparse Polynomial Zonotopes}

\label{subsec:advantages}

As mentioned earlier, one of the main advantages of \SPZsl is that they reduce the dependency problem in Alg. \ref{alg:reach}. We demonstrate this with a short example: 

\begin{example}
We consider the one-dimensional system $\dot x = f(x) = -x + x^2$, the initial set $\mathcal{R}(0) = \{\alpha_1 | \alpha_1 \in [-1,1] \}$, and the time step size $\Delta t = 1$. Computation of the Taylor expansion at $z^* = 0$ in Line~\ref{line:taylor} of Alg.~\ref{alg:reach} results in the parameter values $w=f(z^*) = 0$, $A= \frac{\partial f}{\partial x} |_{z^*}=-1$, and $D= \frac{\partial^2 f}{\partial x^2} |_{z^*}=2$. The quadratic map in Line \ref{line:quadMap} evaluates to $\frac{1}{2} \operator{sq}(D,\mathcal{R}(0)) = \{\alpha_1^2 | \alpha_1 \in [-1,1] \}$ for \SPZs. On the other hand, if we use zonotopes, then the quadratic map has to be over-approximated with $\frac{1}{2} \operator{sq}(D,\mathcal{R}(0)) = \{0.5 + 0.5 \alpha_2 | \alpha_2 \in [-1,1] \}$. Furthermore, the exact addition as defined in Prop.~\ref{prop:exactAddition} is not possible for zonotopes. Therefore, the sets $\mathcal{F}_1$ and $\mathcal{F}_2$ in \eqref{eq:post} have to be added using the Minkowski sum, which results in an additional over-approximation due to the loss of dependency. With zonotopes, we obtain for \eqref{eq:post}
\begin{equation*}
	\begin{split}
		\mathcal{F}_1 \oplus \mathcal{F}_2 = \big \{ & 0.368 \alpha_1 + 0.632(0.5 + 0.5 \alpha_2)   \\
		& |~ \alpha_1,\alpha_2 \in [-1,1] \big \} = [-0.368,1].
		\end{split}
\end{equation*}
With \SPZs, however, we obtain the exact set
\begin{equation*}
	\begin{split}
		\mathcal{F}_1 \boxplus \mathcal{F}_2 = \big \{ & 0.368 \alpha_1 + 0.632 \alpha_1^2 ~|~ \alpha_1 \in [-1,1] \big \} \\
		& = [-0.054,1].
		\end{split}
\end{equation*}
\end{example}
Using zonotopes for reachability analysis therefore leads to a significant over-approximation error in each time step. A similar problem occurs with the polynomial zonotope representation from \cite{Althoff2013a}, since this requires limiting the maximum polynomial degree in advance.

% Hybrid systems

\subsection{Hybrid Systems}

In reachability analysis for hybrid systems, the main difficulty is the calculation of the intersection between the reachable set and the guard sets. For \SPZs, three different strategies exist: 
\begin{enumerate}
 	\item We calculate the intersection with a zonotope over-approximation of the \SPZ . By doing so, it is possible to directly apply the well-developed techniques for the computation of guard intersections with zonotopes, like e.g., the ones from \cite[Sec.~VI]{Althoff2011f} or \cite{Girard2008}.
 	\item We calculate the intersections with the guard sets by using the guard-mapping approach in \cite{Althoff2012a} or the time-scaling approach in \cite{Bak2017c}. Both methods require basic set operations only and can therefore by applied to \SPZs. 
 	\item We apply the method in \cite{Kochdumper2020} which tightly encloses the intersection of the reachable set with guard sets represented by nonlinear level sets with an \SPZ. 
\end{enumerate}
Which method performs best depends on the system.

% NUMERICAL EXAMPLES  ---------------------------------------

\section{Numerical Examples}
\label{sec:NumericalExamples}

In this section we demonstrate the improvements to reachability analysis due to using \SPZsl on four benchmark systems. The computations for our approach are carried out in MATLAB on a 2.9GHz quad-core i7 processor with 32GB memory. Our implementation of \SPZsl will be made publicly available with the next release of the CORA toolbox \cite{Althoff2015a}.

\begin{figure}
\begin{center}
	\includegraphics[width = 0.435 \textwidth]{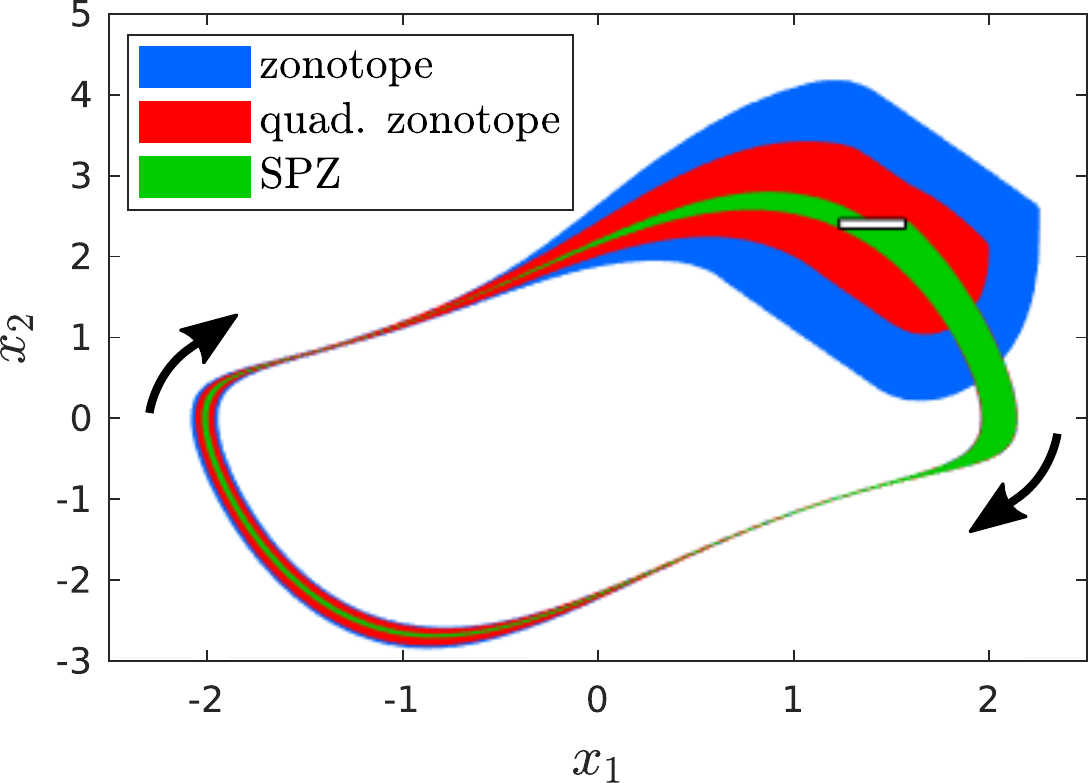}
	\caption{Reachable set of the Van-der-Pol oscillator calculated with different set representations. The initial set is depicted in white with a black border.}
	\label{fig:vanDerPol}
	\end{center}
\end{figure}

\subsection{Van-der-Pol Oscillator}

The system considered first is the Van-der-Pol oscillator taken from the 2019 ARCH competition \cite[Sec.~3.1]{ARCH19nonlinear}:

\begin{equation*}
	\begin{split}
		& \dot x_1 = x_2 \\
		& \dot x_2 = (1-x_1^2) x_2 - x_1.
	\end{split}
\end{equation*} 
For this system, we compare the results for the computation of the reachable set with Alg. \ref{alg:reach} using zonotopes, the quadratic zonotopes from \cite{Althoff2013a}, and our \SPZl representation. We consider the initial set $x_1 \in [1.23,1.57]$ and $x_2 \in [2.34,2.46]$, and use the parameter values $\Delta t = 0.005$ s, $\rho_d = 50$, $\lambda = 0.1$, $\mu_d = 0.01$, and $p_d = 100$. For a fair comparison, we use the same parameter values for every set representation. 

The resulting reachable sets are shown in Fig.~\ref{fig:vanDerPol}. It is clearly visible that the stability of the limit cycle can only be verified with \SPZsl when sets are not split. The computation time is $9.33$ seconds for zonotopes, $13.38$ seconds for quadratic zonotopes, and $16.52$ seconds for \SPZs.

\begin{figure}
\begin{center}
	\includegraphics[width = 0.45 \textwidth]{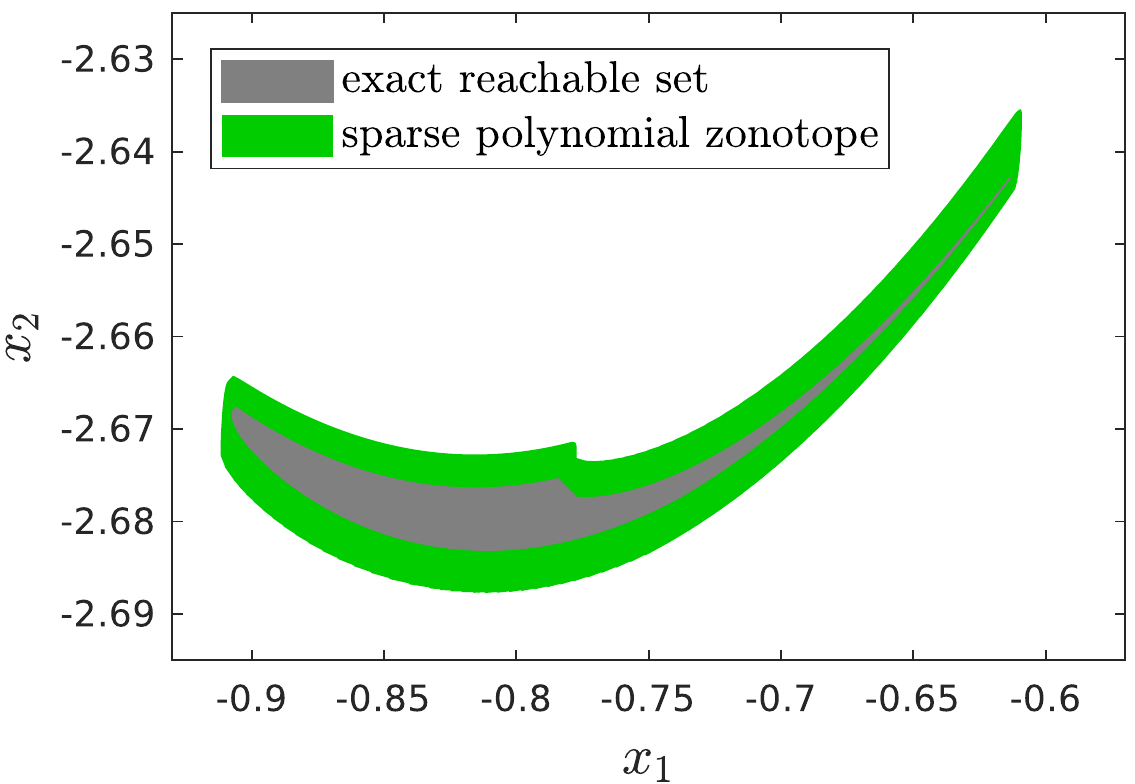}
	\caption{Comparison of the exact reachable set of the Van-der-Pol oscillator after $t=3.15$ seconds with the reachable set over-approximation calculated with \SPZs.}
	\label{fig:vanDerPolSingleSet}
	\end{center}
\end{figure}

An impression on how tight the reachable set can be over-approximated with \SPZsl is provided in Fig.~\ref{fig:vanDerPolSingleSet}, where the reachable set after $t = 3.15$ seconds computed with a time step size of $\Delta t = 0.0001$ s and a maximum volume ratio of $\mu_d = 0.001$ is compared to the exact reachable set of the system. The figure also demonstrates how well the \SPZl approximates the shape of the exact reachable set.

\subsection{Drivetrain}

For the second numerical example, we examine a drivetrain, which is also a benchmark from the ARCH 2019 competition \cite[Sec.~3.3]{ARCH19linear}. We consider the case with two rotating masses, resulting in a system dimension of $n=11$. The model is a hybrid system with linear dynamics. However, we apply the novel approach from \cite{Bak2017c} for calculating the intersections with guard sets, which is based on time-triggered conversion of guards and results in a significant nonlinearity due to the time-scaling process. The initial set is given by $\mathcal{R}(0) = 0.5 (\mathcal{X}_0 - \operator{center}(\mathcal{X}_0)) + \operator{center}(\mathcal{X}_0)$, where $\mathcal{X}_0$ is defined as in \cite[Sec.~3.3]{ARCH19linear}, and we consider the same extreme acceleration maneuver as in \cite[Sec.~3.3]{ARCH19linear}. As a specification, we require that the engine torque after $1.5$ seconds is at least $59 Nm$, which can be formally specified as $T_m \geq 59Nm ~ \forall t \geq 1.5 s$. 

\begin{figure}[h]
\begin{center}
	\includegraphics[width = 0.48 \textwidth]{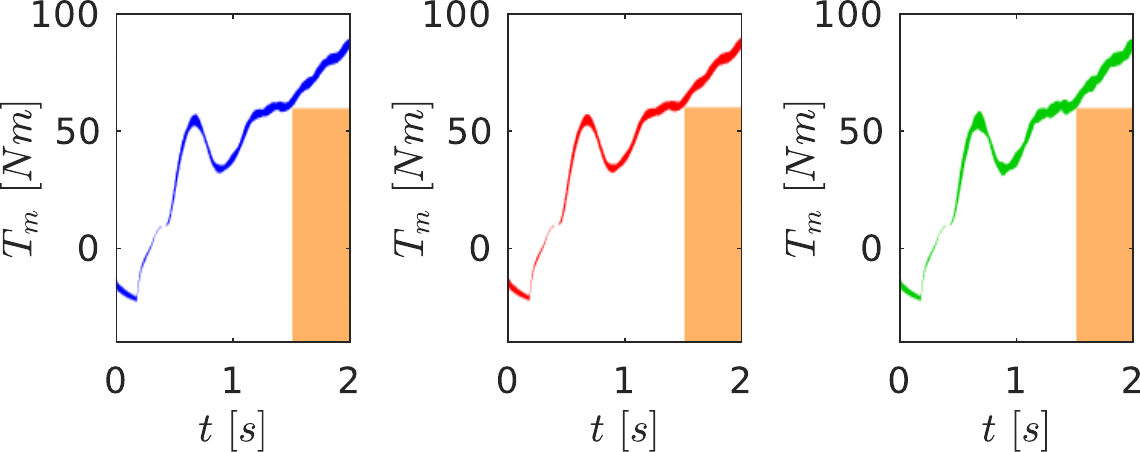}
	\caption{Reachable sets for the drivetrain benchmark calculated with zonotopes (left), quadratic zonotopes (middle), and \SPZsl (right). The forbidden set defined by the specification is depicted in orange.}
	\label{fig:drivetrain}
	\end{center}
\end{figure}

The results for the drivetrain model are shown in Fig.~\ref{fig:drivetrain}. We explicitly considered the possibility of splitting the reachable sets along the largest generator vector so that the specification could be verified with all set representations. However, splitting sets prolongs the computation time: with quadratic zonotopes, the verification took $93$ seconds, and $221$ seconds with zonotopes. Only with \SPZsl was it possible to verify the specification without splitting, resulting in a computation time of $15$ seconds, which is six times faster than with quadratic zonotopes and more than $14$ times faster than with zonotopes. Compared to other non-zonotopic set representations, the speed-up is even larger.

\subsection{Spacecraft Rendezvous}

As a third numerical example we consider the docking-maneuver of a spacecraft taken from the ARCH 2019 competition \cite[Sec.~3.4]{ARCH19nonlinear}. The model is a hybrid system with $n=4$ states and nonlinear dynamics. The three discrete modes are \textit{approaching}, \textit{rendezvous attempt}, and \textit{aborting}. We consider the same initial set and the same specifications as in \cite[Sec.~3.4]{ARCH19nonlinear}. The specifications are that in mode \textit{rendezvous attempt} the spacecraft is located inside the line-of-sight cone and the absolute velocity stays below 3.3 m/min. Furthermore, in mode \textit{aborting} the spacecraft should not collide with the space station.

We apply Alg. \ref{alg:reach} with the parameter values $\Delta t = 0.2$ min (mode \textit{approaching} and \textit{abortion}), $\Delta t = 0.05$ min (mode \textit{rendezvous attempt}), $\rho_d = 10$, $\lambda = 0.1$, $\mu_d = 1$, and $p_d = 10$. To calculate the intersection between the reachable set and the guard sets we use the method in \cite{Kochdumper2020}. The resulting reachable satisfies all specifications. To compare the performance of \SPZsl with other reachability tools we consider the results from the ARCH 19 competition \cite{ARCH19nonlinear}. The comparison in Tab.~\ref{tab:compTime} shows that using \SPZsl resulted in the smallest computation time.

\begin{table}
\begin{center}
\caption{Computation times for the spacecraft rendezvous benchmark. The results for the different tools are taken from \cite[Tab.~4]{ARCH19nonlinear}. The computation times are measured on the machines of the participants (see \cite[Appendix~A]{ARCH19nonlinear}).}
\label{tab:compTime}
\begin{tabular}{ l c c c}
 \toprule
 \textbf{Tool} & \textbf{Comp. Time [s]} & \textbf{Set Rep.} & \textbf{Language} \\ \midrule 
 Ariadne \cite{Benvenuti2014} & $172$ & Taylor models & C++ \\
 CORA \cite{Althoff2015a} & $11.8$ & Zonotopes & MATLAB \\ 
 DynIbex \cite{Sandretto2016b} & $294$ & Zonotopes & C++ \\
 Flow* \cite{Chen2013} & $18.7$ & Taylor models & C++ \\
 Isabelle/HOL \cite{Immler2015} & $295$ & Zonotopes & SML \\
 Our approach & $10.1$ & \SPZs & MATLAB \\
 \bottomrule 
\end{tabular}
\end{center}
\end{table}

\subsection{Transcriptional Regulator Network}

To demonstrate the scalability of our approach, we consider the benchmark in \cite[Sec. VIII.D]{Maiga2015} describing a transcriptional regulator network with $N$ genes. For a network with $N$ genes the system has $n = 2N$ dimensions. We consider the case without artificial guard set so that the benchmark represents a continuous nonlinear system with uncertain inputs. Furthermore, we consider the same initial set, time horizon, and set of uncertain inputs as in \cite[Sec. VIII.D]{Maiga2015}.

We compute the reachable set with \SPZsl using Alg.~\ref{alg:reach} with the parameter values $\Delta t = 0.1$ min, $\rho_d = 10$, $\lambda = 0.1$, $\mu_d = 1$, and $p_d = 50$. The reachable set is visualized in Fig.~\ref{fig:bioNet}, and the computation times for different system dimensions are listed in Tab.~\ref{tab:bioNet}. Even for a system dimension of $n=48$ the computation of reachable set with \SPZsl takes only $122$ seconds, which demonstrates how well our approach scales with the system dimension.

\begin{table}[h]
\begin{center}
\caption{Computation times in seconds for the transcriptional regulator network for different system dimensions.}
\label{tab:bioNet}
\begin{tabular}{l c c c c}
 \toprule
 \textbf{System Dimension} & $\mathbf{n = 12}$ & $\mathbf{n = 24}$ & $\mathbf{n = 36}$ & $\mathbf{n = 48}$ \\ \midrule 
 \textbf{Computation Time [s]} & 6 & 20 & 54 & 122 \\
 \bottomrule 
\end{tabular}
\end{center}
\end{table}

\begin{figure}[h]
\begin{center}
	\includegraphics[width = 0.48 \textwidth]{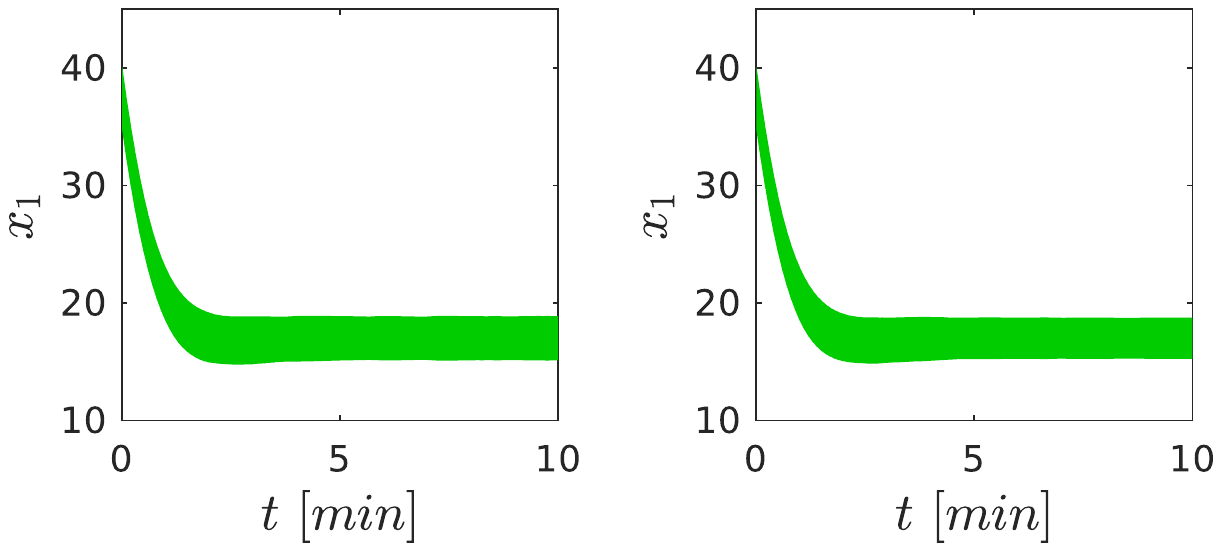}
	\caption{Reachable set of the transcriptional regulator network for the system dimensions $n=12$ (left) and $n=48$ (right).}
	\label{fig:bioNet}
	\end{center}
\end{figure}

% CONCLUSIONS ---------------------------------------------

\section{Conclusions}

We have introduced \textit{sparse polynomial zonotopes}, a new non-convex set representation. The sparsity results in several advantages compared to previous representations of polynomial zonotopes: sparse polynomial zonotopes enable a compact representation of sets, they are closed under all relevant set operations, and all operations have at most polynomial complexity. The fact that sparse polynomial zonotopes are a generalization of several other set representations like Taylor models, polytopes, and zonotopes further substantiates the relevance of this new representation.

The main application for sparse polynomial zonotopes is reachability analysis for nonlinear systems. Numerical examples demonstrate that using sparse polynomial zonotopes results in much tighter over-approximations of reachable sets compared to using zonotopes or quadratic zonotopes. Due to the improved accuracy, splitting can be avoided, resulting in a significant reduction of the computation time, since splitting of sets results in an exponential number of sets to be propagated with respect to the system dimension.

\bibliography{kochdumper,cpsGroup}
\bibliographystyle{plain}

\vskip 0pt plus -30fil

\begin{IEEEbiography}[{\includegraphics[width=1in,height=1.25in,clip,keepaspectratio]{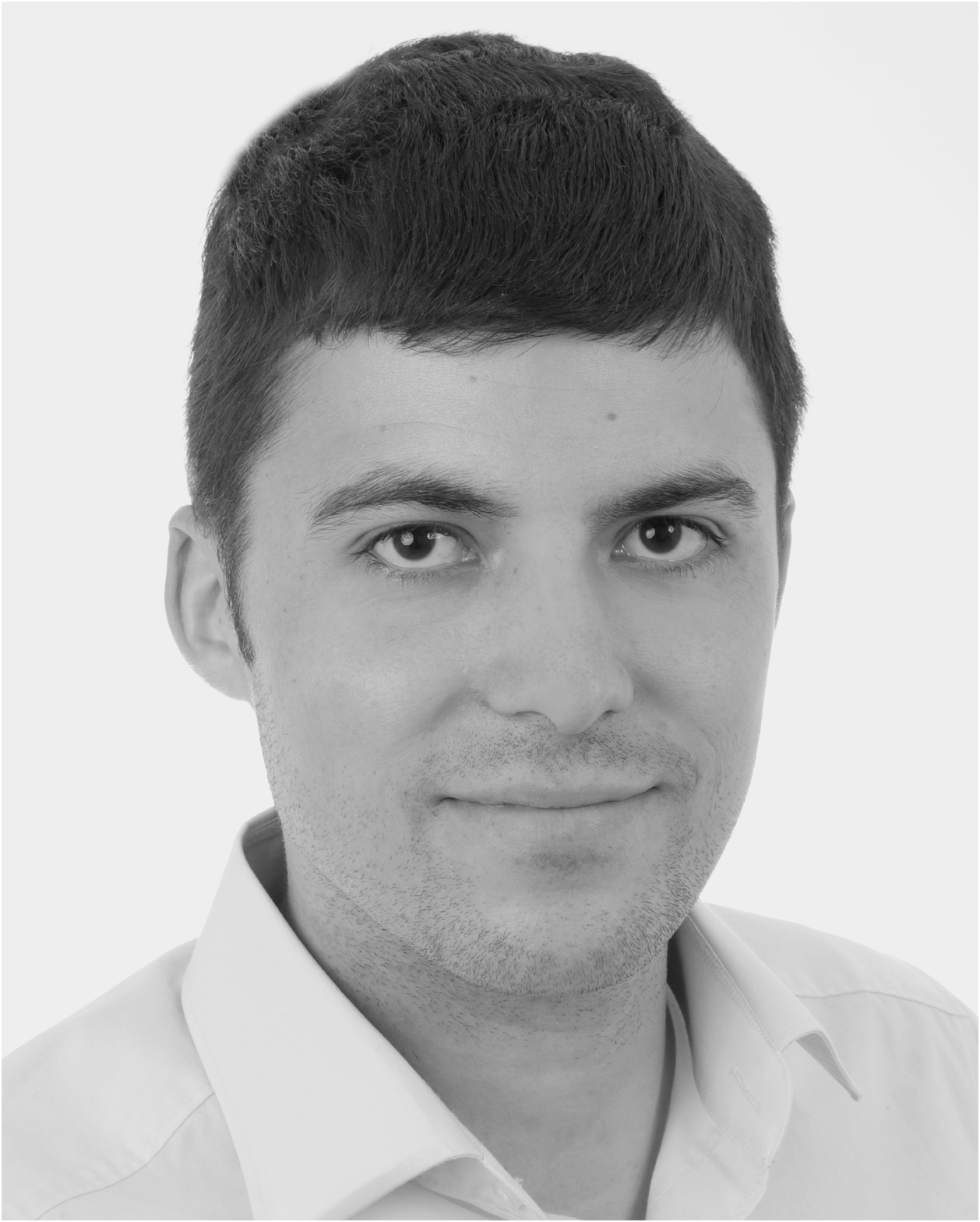}}]{Niklas Kochdumper} received the B.S. degree in Mechanical 
Engineering in 2015 and the M.S. degree in Robotics, Cognition and Intelligence in 2017, both from
Technische Universit\"at M\"unchen, Germany. He is currently pursuing the Ph.D. 
degree in computer science at Technische Universit\"at M\"unchen, Germany. His research interests include formal verification of continuous and hybrid systems, reachability analysis, computational geometry, controller synthesis and electrical circuits.
\end{IEEEbiography}

\vskip 0pt plus -30fil

\begin{IEEEbiography}[{\includegraphics[width=1in,height=1.25in,clip,keepaspectratio]{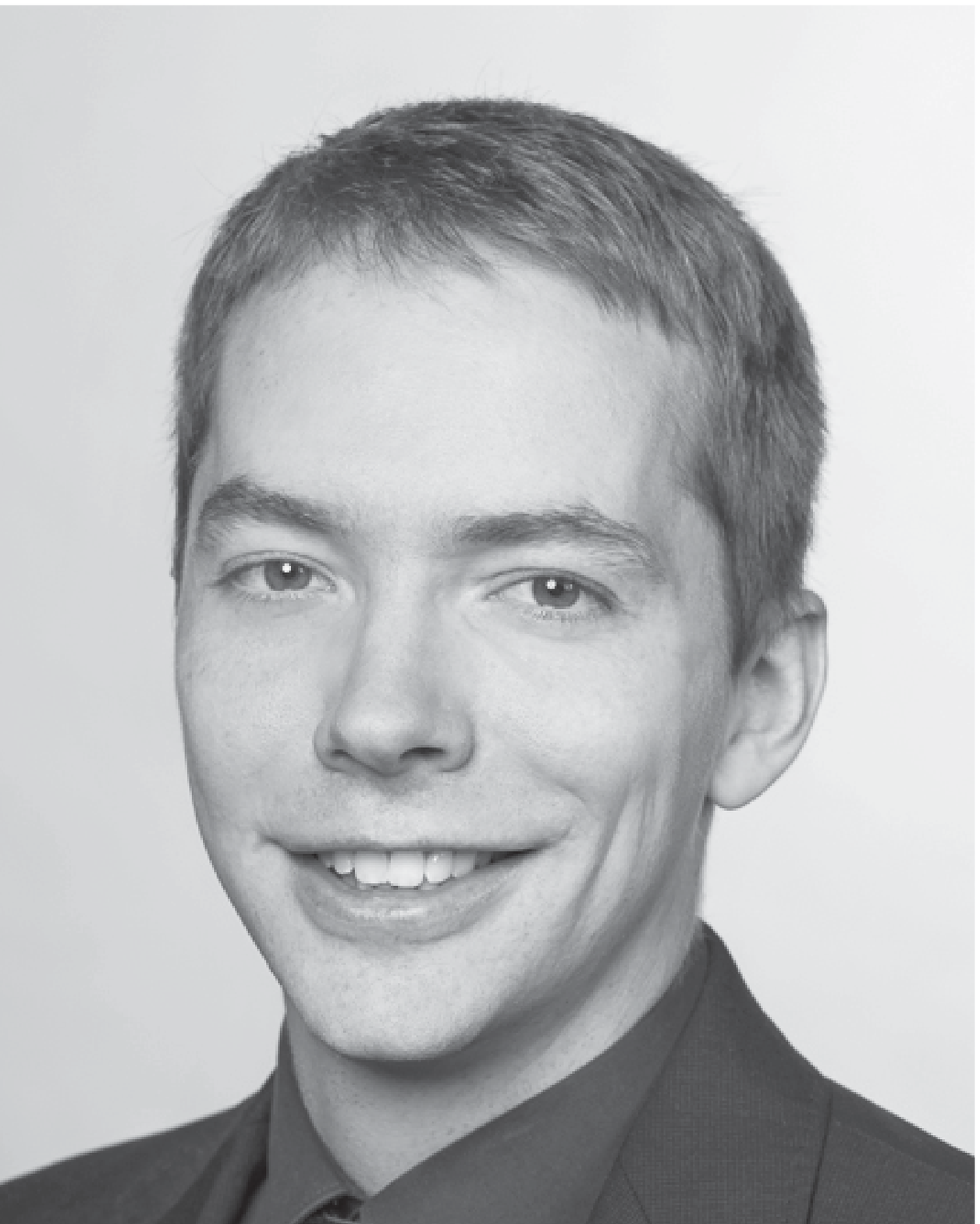}}]
{Matthias Althoff} is an associate professor in computer science at Technische Universit\"at M\"unchen, Germany. He received his diploma engineering degree in Mechanical
Engineering in 2005, and his Ph.D. degree in Electrical Engineering in
2010, both from Technische Universit\"at M\"unchen, Germany.
From 2010 to 2012 he was a postdoctoral researcher at Carnegie Mellon University,
Pittsburgh, USA, and from 2012 to 2013 an assistant professor at Technische Universit\"at Ilmenau, Germany. His research interests include formal verification of continuous and hybrid systems, reachability analysis, planning algorithms, nonlinear control, automated vehicles, and power systems.
\end{IEEEbiography}

\end{document}